\documentclass[a4paper]{SCIYA2023enOL}

\usepackage{bbm}
\usepackage{array}
\usepackage{multirow}
\usepackage{makecell}
\usepackage{threeparttable}
\usepackage{graphicx}
\usepackage[shortlabels,inline]{enumitem}

\newcommand{\Tr}{\mathrm{Tr}}
\newcommand{\rank}{\mathrm{rank}}

\providecommand{\C}{}
\renewcommand{\C}{\mathcal{C}}

\newcommand{\BF}{\mathbb{F}}

\newcommand{\BZ}{\mathbb{Z}}

\DeclareMathOperator{\disc}{disc}
\DeclareMathOperator{\Sym}{Sym}

\DeclareMathOperator{\GL}{GL}
\DeclareMathOperator{\Gal}{Gal}

\DeclareMathOperator{\wt}{wt}

\SetEnumerateShortLabel{a}{\textup{(\alph*)}}
\SetEnumerateShortLabel{A}{\textup{(\Alph*)}}
\SetEnumerateShortLabel{1}{\textup{(\arabic*)}}
\SetEnumerateShortLabel{i}{\textup{(\roman*)}}
\SetEnumerateShortLabel{I}{\textup{(\Roman*)}}

\begin{document}
\ensubject{Coding theory}

\ArticleType{ARTICLES}
\Year{2026}
\Month{January}
\Vol{69}
\No{1}
\BeginPage{1}
\DOI{Science China Mathematics Manuscript for review}

\title[]{Further Study of ${p}$-Ary Self-Orthogonal Codes From Square Functions}
{Further Study of ${p}$-Ary Self-Orthogonal Codes From Square Functions}

\author[1]{Tingting Wu}{wutingting58@163.com}
\author[1,$\ast$]{Xianhong Xie}{{xianhxie@ahau.edu.cn}}

\AuthorMark{Wu T T}
\AuthorCitation{Wu T T, Xie X H}

\address[1]{School of Artificial Intelligence, Anhui Agricultural University, Hefei {\rm 230036}, Anhui, China}

\abstract{In this paper, we investigate a family of trace codes over finite fields of odd characteristic defined by two quadratic trace conditions. By analyzing the quadratic forms associated with these trace conditions and evaluating the corresponding character sums, we obtain the parameters and complete weight distributions of the codes. We further exploit the natural multiplicative orbit structure of the defining sets to obtain punctured codes and determine their parameters, weight distributions and dual distances. Moreover, we establish self-orthogonality criteria and investigate divisibility and projectivity of the constructed codes. These results give a unified treatment for arbitrary odd prime fields and for extension degrees of both parities.}



\maketitle

\section{Introduction}

Linear codes over finite fields are fundamental objects in coding theory, with applications in communications, cryptography, and combinatorics~\cite{1,2}. In particular, linear codes with explicitly determined parameters and weight distributions have attracted considerable attention because of their applications in secret sharing schemes and combinatorial designs~\cite{3,4,5}. Among various classes of linear codes, self-orthogonal codes are of particular interest. A linear code $\C$ over $\mathbb{F}_p$ is called self-orthogonal if $\C\subseteq \C^{\perp}$, where $\C^{\perp}$ denotes its Euclidean dual. Self-orthogonal codes have been extensively studied due to their connections with quantum error correction, lattices and linear complementary dual (LCD) codes~\cite{6,7,8,9,10,11}. Therefore, constructing self-orthogonal codes with good parameters and explicitly determining weight distributions remains an important topic in coding theory.

The defining-set construction provides a useful framework for constructing linear codes over finite fields. For a subset \(D\) of \(\mathbb F_q\), the trace code associated with \(D\) is defined by

\[
\C_D=
\left\{
\left(\Tr(cx)\right)_{x\in D}
:c\in\mathbb F_q
\right\},
\]
the set $D$ is called the defining set of $\C_D$. When the defining set is specified by algebraic functions or trace conditions, the parameters and weight distributions of the resulting codes can often be determined by exponential sums, multiplicative character sums, and quadratic forms~\cite{12,13,14}.

The defining-set construction has also been used to obtain self-orthogonal codes. Several families of self-orthogonal codes have been constructed from weakly regular bent functions and related special functions~\cite{10,11}. Other approaches include constructions based on $p$-divisible codes~\cite{21} and vectorial dual-bent functions~\cite{23}. Trace codes with self-orthogonality properties were studied in~\cite{22}, while Zhang \emph{et al.}~\cite{24} constructed several families of $q$-ary self-orthogonal codes directly via the defining-set construction. These works demonstrate that appropriately chosen defining sets can provide an effective way to simultaneously control the weight distribution and orthogonality properties of the resulting codes.

A construction particularly relevant to this paper was recently given by Xiang and Tang~\cite{15}. They considered ternary self-orthogonal codes associated with the square function \(f(x)=x^2\) and two trace conditions. Their construction shows that imposing two quadratic trace conditions on the defining set can produce ternary self-orthogonal codes with flexible parameters and relatively few nonzero weights. This naturally leads to the following question: can this double-trace construction be extended from the ternary case to arbitrary odd prime fields? If so, how do the arithmetic properties of the defining parameter affect the parameters, weight distributions and structural properties of the resulting codes?

In this paper, we investigate this question by considering a family of \(p\)-ary trace codes defined by two quadratic trace conditions. Define
\[
D_a=
\left\{
x\in\mathbb F_q^*:
\\ \Tr(x^2)=0,\quad
\Tr(ax^2)=0
\right\},
\]
and consider the trace code

\[
\C_{D_a}=
\left\{
\left(\\ \Tr(cx)\right)_{x\in D_a}
:c\in\mathbb F_q
\right\}.
\]

The main contributions of this paper are summarized as follows.
\begin{enumerate}
\item We extend the double-trace construction based on the square function in \cite{15} from the ternary even-degree setting to arbitrary odd primes and both odd and even extension degrees.

\item We determine the complete weight distributions of the constructed codes by evaluating quadratic character sums and enumerating solutions to systems of quadratic trace equations.

\item We clarify the dependence of the code parameters and weight distributions on the intermediate field \(\mathbb F_p(a)\), the degree
\[
d=[\mathbb F_p(a):\mathbb F_p],
\]
and the quotient \(e=m/d\), with separate treatments according to the parity of \(m\) and \(e\).

\item We exploit the \(\mathbb F_p^*\)-orbit structure of the defining set to construct the associated punctured codes and determine their parameters, weight distributions and dual distances.

\item We establish self-orthogonality criteria and investigate divisibility and projectivity for the constructed codes and their punctured codes.
\end{enumerate}

\section{Preliminaries}
\subsection{Notations}
Throughout this paper, we adopt the following notations:
\begin{itemize}
	\item $p$ is an odd prime and $q=p^m>1$;
	\item $\zeta_{p}=\exp(2\pi\mathrm{i}/p)$ is a primitive $p$-th root of unity, where $\mathrm{i}=\sqrt{-1}$;
	\item $\BF_{q}$ is a finite field with $q$ elements and $\BF_{q}^{*}:=\mathbb{F}_{q}-\{0\}$ is the multiplicative group of $\mathbb F_{q}$;
	\item $\Tr:=\Tr_{\mathbb F_q/\mathbb F_p}$ is the absolute trace from $\mathbb F_{q}$ to $\mathbb F_{p}$;
	\item $\eta_q$ and $\eta_p$ denote the quadratic characters of $\BF_{q}^{*}$ and $\BF_{p}^{*}$, respectively, and let $\eta_q(0)=\eta_p(0)=0$;
	\item $\chi_q$ and $\chi_p$ denote the canonical additive characters of $\BF_{q}$ and $\BF_{p}$, respectively;
	\item $G_q:=G_q(\eta_q)$ and $G_p:=G_p(\eta_p)$ denote the corresponding quadratic Gauss sums; whenever $K=\BF_{p^d}$ is used, $\eta_K$, $\chi_K$ and $G_K:=G_K(\eta_K)$ have the analogous meanings, and $q_K:=|K|=p^d$;
	\item For a matrix $A$, $A^T$ is its transpose. In particular, the transpose of a column vector is a row vector;
	\item For a field extension $K/\mathbb F_p$, we write $N_{K/\mathbb F_p}$ for the norm.
\end{itemize}
\subsection{Projective line}
\begin{definition}\label{projective}
	The projective line over $\BF_p$, denoted $\mathbb P^1(\BF_p)$, is the set of one-dimensional subspaces of the two-dimensional vector space $\BF_p^2$. Equivalently,
	\[\mathbb P^1(\BF_p)=(\BF_p^2\setminus\{(0,0)\})/\sim,\] where the equivalence relation $\sim$ is defined by
	\[
	(x_0,x_1)\sim (y_0,y_1)
	\Longleftrightarrow \text{there exists }\lambda\in\BF_p^*\ 
	\text{such that }(x_0,x_1)=\lambda(y_0,y_1).
	\]
\end{definition}

An element of $\mathbb P^1(\BF_p)$ is called a point and is written in homogeneous coordinates as $[x_0:x_1]$ with $(x_0,x_1)\neq(0,0)$ and
\[[x_0:x_1]=[\lambda x_0:\lambda x_1]\ \text{for all } \lambda\in\BF_p^*.\]

Every point with $x_1\neq0$ can be represented uniquely as $[t:1]$, where $t=x_0/x_1\in\BF_p$. The remaining point $[1:0]$ is called the point at infinity and is denoted by $\infty$. Hence there is a bijection \[\mathbb P^1(\BF_p)\leftrightarrow \BF_p\cup \{\infty\},\]
	and the total number of points is
\[\#\mathbb P^1(\BF_p)=p+1.\]

The group of automorphisms of $\mathbb P^1(\BF_p)$ as a projective line is $\operatorname{PGL}_2(\BF_p)$, acting by M\"{o}bius transformations
\[z\mapsto \frac{az+b}{cz+d},\qquad a,b,c,d\in\BF_p,\ ad-bc\neq0,\] with the usual conventions for $z=\infty$.
\subsection{ Quadratic Gauss sums}
\begin{definition}\label{gauss}
	Suppose $\eta_q$ is a quadratic multiplicative character of $\BF_{q}$, the Gauss sum $G_q(\eta_q)$ over $\BF_{q}$ is defined by
	\[G_q(\eta_q)=\sum_{x\in\BF_q}\eta_q(x)\chi_q(x).\]
\end{definition}
The following lemmas can be found in \cite{2,16}.
\begin{lemma} \label{lemma 1}
	Let $p$ be a prime, and $q=p^m$. Then
	\[
	\begin{aligned}
	G_q
	&=(-1)^{m-1}\mathrm{i}^{(\frac{p-1}{2})^2m}p^{m/2}\\
	&=\begin{cases}
	(-1)^{m-1}p^{m/2}, & p\equiv1\pmod4\\
	(-1)^{m-1}\mathrm{i}^mp^{m/2}, & p\equiv3\pmod4.
	\end{cases}
	\end{aligned}
	\]
\end{lemma}
\begin{lemma}\label{exp}
	Let $g(x)=b_2x^2+b_1x+b_0 \in \BF_{q}[x]$, where $b_2\neq 0$. Then
\[
	\sum_{x\in \BF_q}\chi_q(g(x))=\chi_q(b_0-b_1^2(4b_2)^{-1})\eta_q(b_2)G_q.
\]
\end{lemma}	
\begin{lemma}\label{chengkai}\cite[Lemma~3.1]{20}
	Suppose $q$ is an odd prime power and $\mathscr Q(X_1,\ldots,X_n)$ in $n$ variables over $\BF_q$ has rank $r$. Then
	\[\sum_{(x_1,\ldots,x_n)\in\BF_q^n}\chi_q(\mathscr Q(x_1,\ldots,x_n))=\eta_q(\gamma)G_q^rq^{n-r},\]
	where $\gamma$ is a diagonal multiplier of $\mathscr Q$.
\end{lemma}
\section{The code $\C_{D_a}$ from double trace functions}
In this section, we always assume
\[
\begin{aligned}
a&\in\BF_{p^m}\setminus\BF_p,&
d&=[\BF_p(a):\BF_p],\\
K&=\BF_p(a)=\BF_{p^d},&
e&=m/d.
\end{aligned}
\]
For later use, define the two trace quadratic forms on the underlying $\BF_p$-space $V_q=\BF_q$ by
\[
Q_0(x):=\Tr(x^2),\qquad Q_a(x):=\Tr(ax^2).
\]
	Let
	\[
	\phi:\ \BF_q\longrightarrow\mathbb F_p^{\#D_a},\qquad
	c\longmapsto(\Tr(cx))_{x\in D_a},
	\]
	and set $\mathcal H:=\#\ker\phi$.
	When $m$ is even, we have
\[a\in K_0:=\BF_{p^{m/2}}\Longleftrightarrow d\mid\frac{m}{2}\Longleftrightarrow 2\mid e.\]

Set
\[
\begin{aligned}
\theta&=\eta_p(-1),\\
\delta&=\theta^{m/2} &&(2\mid m),\\
\varepsilon_m&=\theta^{(m+3)/2} &&(2\nmid m),\\
\delta_e&=\eta_K(-1)^{(e+3)/2} &&(2\nmid e).
\end{aligned}
\]
  When $m$ is even, Lemma~\ref{lemma 1} gives
 \begin{align}
 G_q&=-\delta p^{m/2}\notag\\
 &=\begin{cases}
 -p^{m/2}, & \mbox{if } p\equiv1\pmod4 \\
 -(-1)^{m/2} p^{m/2}, & \mbox{if }p\equiv3\pmod4.
 \end{cases}\label{me1}
 \end{align}
\begin{lemma}\label{quadratic}
	For any $(u,v)\in\BF_p\times\BF_p\setminus\{(0,0)\}$, take $c=u+va$. Then the quadratic form
	\[Q_c(x):=Q_{u,v}(x)=\Tr((u+va)x^2)=\Tr(cx^2)\]is a non-degenerate quadratic form of rank $m$ over the $m$-dimensional $\BF_p$-space $\BF_q$.
	
	Moreover, let $H_{\rm tr}=(\Tr(e_ie_j))_{i,j}$ be the trace pairing matrix under an arbitrary basis $(e_i)$ of $\BF_q/\BF_p$ ,and let $M_c$  be the symmetric matrix corresponding to the quadratic form $x\mapsto \Tr(cx^2)$. Then the determinant of $Q_{c}$ is given by
	\begin{equation}\label{qua1}
	\det M_c=\det(H_{\rm tr})N_{\BF_q/\BF_p}(c),
	\end{equation}and the quadratic character satisfies
	\begin{equation}\label{qua2}
	\eta_p(\det H_{\rm tr})=(-1)^{m-1}.
	\end{equation}

	In particular, the discriminant of the trace form $x\mapsto \Tr(x^2)$ is a square class  if $m$ is odd, and a nonsquare class if $m$ is even.
\end{lemma}
\begin{proof}
	Since $a\notin\BF_p$, $c=u+va\neq0$ for $(u,v)\neq(0,0)$. The polarization form of $Q_{c}$ is
	\[B_{c}(x,y):=2\Tr(cxy).\]
	By the non-degeneracy of the trace pairing, $B_{c}$ is non-degenerate, hence the rank of $Q_{c}$ is $m$.
	
	Under the basis $(e_i)$, the matrix of the quadratic form is $H_{\rm tr}M_{{\rm mul},c}$, where $M_{{\rm mul},c}$ is the matrix of the multiplication map $x\mapsto cx$. Therefore, \[
	\det(H_{\rm tr}M_{{\rm mul},c})=\det(H_{\rm tr})N_{\BF_q/\BF_p}(c),
	\]which establishes Eq.~\eqref{qua1}.
	
	By Lemma~\ref{exp} and the Hasse-Davenport theorem,
	\begin{equation}\label{1}\sum_{x\in\BF_q}\zeta_p^{\Tr(cx^2)}
	=\eta_q(c)G_q=\eta_p(N_{\BF_q/\BF_p}(c))(-1)^{m-1}G_p^m.\end{equation}
	In addition, since $M_c$ is the symmetric matrix of $Q_{c}$, by Lemma~\ref{chengkai}, we have \begin{equation}\label{2}\sum_{x\in\BF_q}\zeta_p^{\Tr(cx^2)}
	=\eta_p(\det M_c)G_p^m=\eta_p(\det(H_{\rm tr})N_{\BF_q/\BF_p}(c))G_p^m.\end{equation}
	Combining ~\eqref{1} and ~\eqref{2} yields Eq.~\eqref{qua2}.
\end{proof}

\subsection{The case of even $m$}
In $\mathbb P^1(\BF_p)$, we take
\[
\alpha_\infty=1,\qquad \alpha_t=a-t\quad(t\in \BF_p),
\]
	and define
\begin{equation}\label{ha}
h(a)=\sum_{j\in\mathbb P^1(\BF_p)}\eta_q(\alpha_j)
=1+\sum_{t\in \BF_p}\eta_q(a-t).
\end{equation}
Suppose $f_a(X)\in\BF_p[X]$ is the minimal polynomial of $a$, then
\[N_{K/\BF_p}(a-t)=(-1)^df_a(t).\]
Note that for $A\in K^*$,
\[\eta_q(A)=\eta_p(N_{K/\BF_p}(A))^e.\]
	Therefore, we have
\begin{align}
h(a)
&=1+\sum_{t\in \BF_p}\eta_q(a-t)
=1+\sum_{t\in \BF_p}\eta_p(N_{K/\BF_p}(a-t))^e\notag\\
&=1+\sum_{t\in \BF_p}\eta_p\!\left((-1)^df_a(t)\right)^e.\label{ha1}
\end{align}
\begin{lemma}\label{n0}
	Set	\[
	n_0=\#\{x\in \BF_q:\Tr(x^2)=\Tr(ax^2)=0\}.
	\]
	Then
	\[
	n_0=p^{m-2}-\delta(p-1)p^{m/2-2}h(a),\] We further have
	\[n=|D_a|=n_0-1=p^{m-2}-1-\delta(p-1)p^{m/2-2}h(a).
	\]
	In particular,
	\begin{enumerate}
		\item If $2\mid e$, then $\eta_q|_{K^*}=1$ and
		\[h(a)=p+1,\qquad
		n=p^{m-2}-1-\delta(p^2-1)p^{m/2-2}.\]
		\item If $2\nmid e$, then $2\mid d$ and $h(a)=1+\sum\limits_{t\in \BF_p}\eta_p(f_a(t))$ generally depends on the minimal polynomial $f_a(X)$. Hence,
		\[
		n=p^{m-2}-1-\delta(p-1)p^{m/2-2}h(a).
		\]
	\end{enumerate}
\end{lemma}
\begin{proof}
	Note that \begin{align*}
	n_0&=\frac1{p^2}
	\sum_{u,v\in\BF_p}\sum_{x\in \BF_q}
	\chi_q((u+av)x^2)\\
	&=p^{m-2}+\frac1{p^2}
	\sum  _{(u,v)\in\BF_p\times\BF_p\setminus\{(0,0)\}} \sum_{x\in \BF_q}
	\chi_q((u+av)x^2).
	\end{align*}
	For $(u,v)\neq(0,0)$, the coefficients can be written as $\rho \alpha_j$ among $p+1$ projective lines, where $\rho\in\BF_p^*$. If $2\mid m$, then
	\[\eta_q(\rho)=\eta_p(\rho)^m=1.\]
	By Eq.~  \eqref{me1} , we have
	\[
	\begin{aligned}
	n_0&=p^{m-2}+\frac{p-1}{p^2}G_qh(a)=p^{m-2}-\delta(p-1)p^{m/2-2}h(a),\\
	n&=\#D_a=n_0-1.
	\end{aligned}
	\]
	Suppose \[a'=\frac{ra+s}{ta+u}\] is a well-defined transformation in $\operatorname{PGL}_2(\BF_p)$, then
	\[
	\begin{aligned}
	\langle1,a'\rangle_{\BF_p}
	&=(ta+u)^{-1}\langle1,a\rangle_{\BF_p},\\
	h(a')&=\eta_q((ta+u)^{-1})h(a).
	\end{aligned}
	\]
	Hence, if $2\mid e$, $h(a')=h(a)$, and if $2\nmid e$, the sign can change, the function $h(a)$ cannot be said to depend only on the value of the unweighted $\operatorname{PGL}_2(\BF_p)$-orbit.
\end{proof}
For $\beta\in K$, define
\begin{equation}\label{Za}
Z_a(\beta)=
\mathrm{1}_{\{\Tr_{K/\BF_p}(\beta)=0\}}
+\sum_{t\in \BF_p}
\eta_q(a-t)\,
\mathrm{1}_{\{\Tr_{K/\BF_p}(\beta/(a-t))=0\}}.
\end{equation}Clearly, $Z_a(0)=h(a)$.
For any   $b\in\BF_q$ , we take
\begin{equation*}
\beta(b)=\Tr_{\BF_q/K}(b^2).
\end{equation*}
\begin{lemma}\label{bb}
	Let $\eta_K$ denote the quadratic character of $K$, $q_K=p^d$ and
	\[
	R_e(\beta)=
	\#\{b\in\BF_q:\Tr_{\BF_q/K}(b^2)=\beta\},  \qquad  \beta\in K.
	\]
	Then
	\begin{itemize}
		\item If $2\mid e$, then
		\[
		\begin{aligned}
		R_e(0)&=q_K^{e-1}-\eta_K(-1)^{e/2}(q_K-1)q_K^{e/2-1},\\
		R_e(\beta)&=q_K^{e-1}+\eta_K(-1)^{e/2}q_K^{e/2-1}
		\quad(\beta\ne0).
		\end{aligned}
		 \]
		\item If $2\nmid e$, then
		 \[
		\begin{aligned}
		R_e(0)&=q_K^{e-1},\\
		R_e(\beta)&=q_K^{e-1}+\delta_e\eta_K(\beta)q_K^{(e-1)/2}
		\quad(\beta\ne0).
		\end{aligned}
		\]
	\end{itemize}
\end{lemma}
\begin{proof}
	  By additive-character orthogonality over $K$,
	 \begin{align*}
	R_e(\beta)
	&=  \frac{1}{q_K} \sum_{y\in K} \chi  _K ( -y\beta)
	 \sum  _{b\in\BF_q} \chi_q(  yb^2 ) \\
	&=  q_K^{e-1} +  \frac{G_q}{q_K}
	 \sum_{y\in K^*}\chi_K(  -y \beta)\eta_K(y)^e  .
	 \end{align*}
	If   $e$ is even , then $\eta_K(y)^e=1$  and
		$G_q=-\eta_K(-1)^{e/2}q_K^{e/2}$ by Hasse-Davenport over $\BF_q/K$.
		This gives the stated values for $\beta=0$ and $\beta\ne0$.

	If $e$ is odd ,  then $\eta_K(y)^e=\eta_K(y)$  and
		$G_q=G_K^{e}$. For $\beta=0$, 
		\[\sum_{y\in K^*}\chi_K(  -y \beta)\eta_K(y)^e =\sum_{y\in K^*}\eta_K(y)=0.\] For $\beta\ne0$,
	  \[
	\begin{aligned}
	\frac{G_q}{q_K}
	\sum_{y\in K^*}\chi_K(-y\beta)\eta_K(y)
	&=\frac{G_K^{e+1}}{q_K}\eta_K(-\beta)\\
	&=\delta_e\eta_K(\beta)q_K^{(e-1)/2}.
	\end{aligned}
	\]
	 This thus proves the lemma.
	 \end{proof}
\begin{lemma}\label{three}
	For any   $b\in\BF_q^*$, set  \[
	N_a(b)=\#\{x\in   \BF_q :
	\Tr(x^2)=\Tr(ax^2)=\Tr(bx)=0\}.
	\]
	Then\[
	N_a(b)=p^{m-3}
	-\delta(p-1)p^{m/2-2}Z_a(\beta(b)).
	\]
\end{lemma}
\begin{proof}
	We have
	\[N_a(b)=\frac1{p^3}\sum_{u,v,s\in \BF_p}
	\sum_{x\in \BF_q}
	\chi_q((u+va)x^2+sbx).\]
	
	If $(u,v)=(0,0)$, then
	\begin{equation}\label{00}
	\sum_{s\in \BF_p}\sum_{x\in \BF_q}\chi_q(sbx)
	=
	\begin{cases}
	p^{m+1},&b=0,\\
	p^m,&b\ne0.
	\end{cases}
	\end{equation}
	Let $b\neq0$ and $c=u+va\neq0$, by Lemma~\ref{exp},
	\[
	\sum_x\chi_q(cx^2+sbx)
	=\eta_q(c)G_q\,
	\chi_q\!\left(-\frac{s^2b^2}{4c}\right).
	\]
	Take $c=\rho \alpha_j$, where $\rho\in\BF_p^*$. Then $\eta_q(\rho)=1$ and
	\begin{equation}\label{sq}
	\begin{aligned}
	&\sum_{\rho\in\BF_p^*}\sum_{s\in\BF_p}\eta_q(\rho \alpha_j)G_q
	\chi_q\!\left(-\frac{s^2b^2}{4\rho \alpha_j}\right)\\
	&\quad=\eta_q(\alpha_j)G_q
	\sum_{\rho\in\BF_p^*}\sum_{s\in\BF_p}
	\chi_p\!\left(-\frac{s^2}{4\rho}\Tr(\frac{b^2}{\alpha_j})\right).
	\end{aligned}
	\end{equation}
	Clearly, the sum of Eq.~\eqref{sq} $=p(p-1)\eta_q(\alpha_j)G_q$ if $\Tr(\frac{b^2}{\alpha_j})=0$, otherwise, $=0$. Notice that
	\[\Tr (b^2/\alpha_j)
	=\Tr_{K/\BF_p} (\beta(b)/\alpha_j).\]Combining with the part of $b\neq0$ in ~\eqref{00}, we obtain
	\[N_a(b)=p^{m-3}
	-\delta(p-1)p^{m/2-2}Z_a(\beta(b)).\qedhere
	\]
\end{proof}
For $b\in\BF_q$, write
\[
\mathbf c_b=(\Tr(bx))_{x\in D_a}.
\]
By Lemmas~\ref{n0} and ~\ref{three}, we have the following result.
\begin{theorem}\label{th:weight1}
	The code $\C_{D_a}$ has weight $\wt( \mathbf c_0)=0$ and
	\[
	\wt( \mathbf c_b)
	=w_a^{\rm even}(\beta(b))=(p-1)p^{m-3}
	-\delta(p-1)p^{m/2-2}\bigl(h(a)-Z_a(\beta(b))\bigr),\quad b\neq0
	\]
	Obviously, the weights of nonzero codes depend only on $\beta(b)=\Tr_{\BF_q/K}(b^2)=\beta$.
\end{theorem}
\begin{example}
Take $p=3$, $m=6$, and let $a\in\BF_{3^3}\subset\BF_{3^6}$ be a root of
$X^3-X-1\in\BF_3[X]$. Then $K=\BF_3(a)=\BF_{3^3}$, $d=3$, and $e=2$.
An exhaustive computation gives
$n=104$. For $b=1$, $\beta(1)=2, Z_a(2)=1,\wt(\mathbf c_1)=72.$

\end{example}

Set
\[\widehat R_e(\beta)=R_e(\beta)-\mathrm{1}_{\{\beta=0\}}.\]
	Then $\widehat R_e(\beta)=\#\{b\in\BF_q^*:\ \beta(b)=\Tr_{\BF_q/K}(b^2)=\beta\}$. Define
\[\mathcal B_{\rm att}=\{\beta\in K:\widehat R_e(\beta)>0\},\]
	and
\begin{equation}\label{H}
\mathcal H=1+
\sum_{\substack{\beta\in\mathcal B_{\rm att}\\w_a^{\rm even}(\beta)=0}}
\widehat R_e(\beta).
\end{equation}
Thus Eq.~\eqref{H} is an explicit evaluation of $\#\ker\phi$. Since the kernel is a $\BF_p$-vector space, $\mathcal H$ must be a power of $p$. Therefore, we have
\begin{theorem}\label{th:weight2}With the same notations as above.
	\begin{itemize}
		\item The dimension of the code is
		\[\dim\C_{D_a}=m-\log_p \mathcal H.\]
		\item The   weight enumerator is
		  \[
		W_{\C_{D_a}}(z)=\sum_{w\ge0}A_wz^w,
		\]
		 where
		\begin{equation}\label{wei}A_w=
			\frac1{\mathcal H}\left(
			\mathrm{1}_{\{w=0\}}
			+\sum_{\substack{\beta\in\mathcal B_{\rm att}\\w_a^{\rm even}(\beta)=w}}
			\widehat R_e(\beta)
			\right),
			\qquad w\ge0.\end{equation}
		 In particular, if  $w=0$, then the expression inside the parentheses is exactly $\mathcal H$, so $A_0=1$.
		\item All actual non-zero weights are precisely \[\{w_a^{\rm even}(\beta):\beta\in\mathcal B_{\rm att},\ w_a^{\rm even}(\beta)>0\},\]
	and the exact minimum distance is \[d_{\min}
		=\min_{\substack{\beta\in\mathcal B_{\rm att}\\w_a^{\rm even}(\beta)>0}}
		w_a^{\rm even}(\beta).\]
	\end{itemize}
	
\end{theorem}
\begin{example}
Take $p=3$, $m=6$, and let $a\in\BF_{3^3}\subset\BF_{3^6}$ be a root of
$X^3-X-1\in\BF_3[X]$. Then $K=\BF_3(a)=\BF_{3^3}$, $d=3$ and $e=2$.
An exhaustive computation gives
\[
\mathcal H=1,\qquad \dim\C_{D_a}=6,\qquad d_{\min}=54,
\]
and
\[
W_{\C_{D_a}}(z)
=1+52z^{54}+312z^{66}+208z^{72}+156z^{78}.
\]
\end{example}
Next we shall study the injectivity, low-dimensional degeneration and Pless power moment.
\begin{proposition}\label{pro:dim} 
	The following conclusions hold.
	\begin{enumerate}[1]
		\item If $m\geq6$, the map $\phi$ is necessarily injective, hence $\mathcal H=1$ and the dimension is $m$.
		\item If $m=2$, then $d=2,e=1$ and
		\[h(a)=0,\qquad D_a=\emptyset,\]thus the code is the zero code with length and dimension both equal to 0.
		\item If $m=4,d=2,e=2$, then $h(a)=p+1$ and   \[
		D_a=\emptyset.
		\]
		\item If $m=4,d=4,e=1$, then
\begin{equation*}
		\#D_a=(p-1)(p+1-h(a)).
		\end{equation*}The map $\phi$ is injective if and only if $\mathcal H=1$ in ~\eqref{H}. Furthermore, it is injective when $p\geq11$.
		\item When $n=\#D_a$ and $\#\ker(\phi)=\mathcal H$, Eq.~\eqref{wei} satisfies
		\begin{align*}
		\sum_wA_w & =p^m/\mathcal H \\
		\sum_wwA_w & =p^{m-1}(p-1)n/\mathcal H \\
		\sum_ww^2A_w
		& =p^{m-2}(p-1)^2n(n+1)/\mathcal H.
		\end{align*}
	\end{enumerate}
\end{proposition}
\begin{proof}
	(1) For $b\in\BF_q^*$, by Eqs.~\eqref{ha} and ~\eqref{Za}, we know \[
	|h(a)-Z_a(\beta(b))|\le p+1.
	\]
	By Eq.~\eqref{wei}, if $m\geq6$, then \[
	w_a^{\rm even}(\beta(b))\ge
	(p-1)p^{m/2-2}\bigl(p^{m/2-1}-(p+1)\bigr)>0.
	\]
	This implies that no non-zero parameter belongs to the kernel, i.e., $\mathcal H=1$.
	\smallskip
	
	(2) When $m=2$, $f_a$ is an irreducible quadratic polynomial, and \[
	\sum_{t\in\BF_p}\eta_p(f_a(t))=-1.
	\]By Eq.~\eqref{ha1}, $h(a)=0$ and then $n=0$. For $m=4,d=2$, we know
	\[n=p^{m-2}-1-\delta(p^2-1)p^{m/2-2}=p^2-1-(p^2-1)=0.\]
	Hence, the proof of (3) is also completed.
	
	(4) Now let   $m=4,d=4$ . By Lemma~\ref{n0}, then
	\[n=\#D_a=(p-1)(p+1-h(a)).\] In $\mathbb P(V_q)\cong\mathbb P^3_{\BF_p}$, let
	\[
	\mathscr X_a: \
	\Tr(x^2)=
	\Tr(ax^2)=0.
	\] Its $\BF_p$-rational points are exactly $D_a/\BF_p^*$, and then \[
	\#\mathscr X_a(\BF_p)=\frac{\#D_a}{p-1}=p+1-h(a).
	\]
	Below we show the geometric smoothness of $\mathscr X_a$. Take the algebraic closure $\overline{\BF}_p$, and denote \[a_1,a_2,a_3,a_4\] as the four $\BF_p$-conjugates of $a$. Since the finite field extension is separable and the degree of $a$ is $4$, they are distinct pairwise. Via the split isomorphism
	\[
	\BF_q\otimes_{\BF_p}\overline \BF_p\simeq\overline \BF_p^4,
	\]
	we can take coordinates $X_1,\ldots,X_4$ such that the two quadratic equations become
	\[
	Q_0=\sum_{i=1}^4X_i^2=0,\qquad
	Q_a=\sum_{i=1}^4a_iX_i^2=0.
	\]
	
	If $[X_1:\cdots:X_4]\in \mathscr X_a(\overline{\BF}_p)$  is a singular point, the Jacobian matrix has rank less than $2$, so there exists $(u,v)\ne(0,0)$ such that
	\[
	u\nabla Q_0+v\nabla Q_a=0.
	\]
	Since $p$ is odd, this is equivalent to
	\begin{equation}\label{31a}
	(u+va_i)X_i=0\qquad(1\le i\le4).
	\end{equation}
	If $v=0$, then $u\neq0$, which gives all $X_i=0$ from ~\eqref{31a}, this is impossible. If $v\neq0$, since the $a_i$
	are distinct, at most one index $j$ satisfies
	\[u+va_j=0,\]
	thus except for the possible $X_j$, all other coordinates are zero. Furthermore, from $Q_0=0$ we get $X_j^2=0$, which still implies all coordinates are zero, a contradiction. Therefore $\mathscr X_a$ is smooth over $\overline{\BF}_p$.
	
	Thus $\mathscr X_a$ is a smooth complete intersection of two quadric surfaces in $\mathbb P^3_{\BF_p}$.  It is geometrically connected, hence geometrically integral, by the adjunction formula, its degree is $4$ and its arithmetic genus is 1. The Hasse-Weil bound then gives
	\[
	\#\mathscr X_a(\BF_p)\ge p+1-2\sqrt p.
	\]
	When $p\geq11$,
	\[
	p+1-2\sqrt p>4.
	\]
	
	Furthermore, $\mathscr X_a$ is not contained in any hyperplane, its homogeneous ideal is the complete intersection ideal generated by two quadratic equations, which does not contain any non-zero linear homogeneous form. Therefore, if the intersection of any hyperplane with $\mathscr X_a$ is not the whole $\mathscr X_a$, by B\'{e}zout's theorem, it then has at most 4 geometric points, counted with multiplicity. If the map $\phi$ has a non-zero kernel element $b$, then the hyperplane
	\[
	\Tr(bx)=0
	\]
	contains all points of $\mathscr X_a(\BF_p)$. Therefore, once $\#\mathscr X_a(\BF_p)>4$,  these rational points cannot all lie on that hyperplane, and the evaluation map is necessarily injective. This proves the assertion for $p\geq11$.
	
	(5) The Pless power moment equations are the quotient of the parameter space by $\#\ker(\phi)=\mathcal H$. Therefore, the proof is finished.
\end{proof}

\begin{remark}
	When $m=d=4$, we obtain the following results through experiments for $p=3,5,7$.
	\begin{itemize}
		\item All $72$ degree-four parameters in $\BF_{3^4}$, the numbers of orbits are $2,4,6$,  respectively, and the corresponding numbers of parameters are $36, 24, 12$ respectively, the dimension is $2$ when the orbit number is $2$, and the dimension is $4$ for the rest.
		\item All $600$ degree-four parameters in $\BF_{5^4}$, the numbers of orbits are $2,4,6,8,10$,  respectively, and the corresponding numbers of parameters are $100, 160, 240, 80, 20$ respectively, the dimension is $2$ when the orbit number is $2$, and the dimension is $4$ for the rest.
		\item All $2352$ degree-four parameters in $\BF_{7^4}$, the numbers of parameters for orbit numbers $4, 6, 8, 10, 12$ are $504, 840, 336, 504, 168$ respectively, and the evaluation dimension is $4$ for all.
	\end{itemize}
\end{remark}
\subsection{The case for odd $m$}
\begin{lemma}\label{om1}
	Set
	\[n_0=\#\{x\in \BF_{q}:\Tr(x^2)=\Tr(ax^2)=0\}.
	\]
	Then
	\begin{equation}\label{mo2}
	n_0=p^{m-2},\qquad n=\#D_a=p^{m-2}-1.
	\end{equation}
\end{lemma}
\begin{proof}
	\begin{align*}
	n_0&=p^{-2}\sum_{u,v\in\BF_p}
	\sum_{x\in \BF_{q}}\chi_q((u+va)x^2)\\
	&=p^{m-2}+p^{-2}\sum_{(u,v)\in\BF_p\times\BF_p\setminus\{(0,0)\}}
	\sum_{x\in \BF_{q}}\chi_q((u+va)x^2).
	\end{align*}
	Grouping $(u,v)\in\BF_p\times\BF_p\setminus\{(0,0)\}$ by $\mathbb P^1(\BF_p)$, each group can be written as $\lambda(u_0,v_0)$, then
	\[
	G_{q}\eta_{q}(u_0+v_0a)
	\sum_{\lambda\in\BF_p^*}\eta_{p}(\lambda)=0.
	\]
	Hence $n_0=p^{m-2}$.
\end{proof}

By Lemma~\ref{bb}, the fiber sizes of the map $b\mapsto\Tr_{\BF_{q}/K}(b^2)$ are
\begin{equation}\label{mo10}
\begin{aligned}
R_e(0)&=q_K^{e-1},\\
R_e(\beta)&=q_K^{e-1}
+\delta_e\eta_K(\beta)q_K^{(e-1)/2},\quad \beta\ne0.
\end{aligned}
\end{equation}

For $b\neq0$, define
\begin{equation}\label{mo5}
S_a(b)=
\sum_{[u:v]\in\mathbb P^1(\BF_p)}
\eta_{q}(u+va)\,
\eta_{p}\!\left(
\Tr(\frac{b^2}{u+va})
\right)
\end{equation}

Note that $\eta_q(\lambda)=\eta_p(\lambda)^m=\eta_p(\lambda)$ for $\lambda\in\BF_p^*$, if the projective representative is multiplied by $\lambda$, the two quadratic character factors in ~\eqref{mo5} are multiplied by $\eta_p(\lambda)$ and $\eta_p(\lambda^{-1})$ respectively, so the product remains unchanged. This means that $S_a(b)$ is independent of the choice of representative element.
\begin{lemma}\label{om2}
	 For $b\in\BF_q^*$, set
	 \[N_a(b)=\#\{x\in \BF_{q}:
	\Tr(x^2)=\Tr(ax^2)
	=\Tr(bx)=0\}.\]
	Then
	\begin{equation}\label{mo6}
	N_a(b)=p^{m-3}+(p-1)\varepsilon_mp^{(m-5)/2}S_a(b).
	\end{equation}
\end{lemma}
\begin{proof}
	For $c=u+va\neq0$ and $w\in\BF_p$, then
	\[
	\begin{aligned}
	\sum_{x\in \BF_q}\chi_q(cx^2+wbx)
	&=\eta_q(c)G_q\chi_q\!\left(-\frac{w^2b^2}{4c}\right)\\
	&=\eta_q(c)G_q
	\chi_p\!\left(-\frac{w^2}{4}\Tr(\frac{b^2}{c})\right).
	\end{aligned}
	\]
	
	Fixing the projective representative $\alpha_0=u_0+v_0a$, and let $c=\lambda\alpha_0$ and
	\[
	T:=\Tr(b^2/\alpha_0),
	\]
	
	If $T=0$, then
	\[\sum_{w\in\BF_p}\sum_{\lambda\in\BF_p^*}\eta_{q}(\lambda \alpha_0)G_{q}
	=pG_{q}\eta_q(\alpha_0)\sum_{\lambda\in\BF_p^*}\eta_{q}(\lambda)=0.\]
	
	If $T\neq0$, then
	\[
	\sum_{w\in\BF_p}
	\zeta_{p}^{-w^2T/(4\lambda)}
	=\eta_{p}(-T/\lambda)G_{p}
	=\theta\eta_{p}(T)\eta_{p}(\lambda)G_p.
	\] Note that $\eta_{q}(\lambda \alpha_0)=\eta_{p}(\lambda)\eta_{q}(\alpha_0)$, so the total contribution of this projective class is
	 \[
	( p-1)\theta\eta_{q}(\alpha_0)\eta_{p}(T)G_{q}G_p.
	\]
	  Multiplying by the orthogonality factor  $p^{-3}$  and using
		$G_q=G_p^m$, $G_p^2=\theta p$, gives
		$p^{-3}\theta G_qG_p=\varepsilon_mp^{(m-5)/2}$. This  yields Eq.~\eqref{mo6}.
\end{proof}

For $\beta\in K$, write
\begin{equation}\label{mo8}
s_a(\beta)
=\sum_{[u:v]\in\mathbb P^1(\BF_p)}\eta_K(u+va)
\cdot\eta_p\!\left(\Tr_{K/\BF_p}\frac{\beta}{u+va}\right).
\end{equation}Since $u+va\in K$, \[
\Tr
(\frac{b^2}{u+va})
=
\Tr_{K/\BF_p}
\left(
\frac{\Tr_{\BF_{q}/K}(b^2)}{u+va}
\right).
\]
Note that $2\nmid e$, $\eta_{q}(c)=\eta_K(c)^e=\eta_K(c)$, thus
\begin{equation*}
S_a(b)=s_a\!\left(\Tr_{\BF_{q}/K}(b^2)\right).
\end{equation*}
\begin{theorem}\label{th:weight3}
	If $m\geq5$, then $\dim(\C_{D_a})=m$, and
	\begin{equation*}
	\wt(\mathbf c_b)=
	\begin{cases}
	0, & b=0,\\
	(p-1)\bigl[p^{m-3}-\varepsilon_m p^{(m-5)/2}S_a(b)\bigr],
	& b\neq0.
	\end{cases}
	\end{equation*}
	The weight enumerator is
	\begin{equation*}
	W_{\C_{D_a}}(z)
	=1+\bigl(q_K^{e-1}-1\bigr)z^{(p-1)p^{m-3}}
	+\sum_{\beta\in K^*}
	\left(q_K^{e-1}+\delta_e\eta_K(\beta)q_K^{(e-1)/2}\right)
	z^{w_a^{\rm odd}(\beta)}.
	\end{equation*}
	where \begin{equation*}
	w_a^{\rm odd}(\beta)
	=(p-1)p^{m-3}
	-(p-1)\varepsilon_m p^{(m-5)/2}s_a(\beta).
	\end{equation*}
	
	Furthermore, if we take \[
	\begin{aligned}
	\mathcal B_o&=
	\begin{cases}
	\{\beta\in K^*:\eta_K(\beta)=1\},&e=1,\\
	K,&e\ge3,
	\end{cases}\\
	M_a&=\max_{\beta\in\mathcal B_o}\varepsilon_ms_a(\beta).
	\end{aligned}
	\]
	then the exact minimum distance of $\C_{D_a}$ is
	\begin{equation}\label{mo13}
	d_{\min}
	=(p-1)p^{m-3}
	-(p-1)p^{(m-5)/2}M_a.
	\end{equation}
\end{theorem}
\begin{proof}
	By Lemmas~\ref{om1} and ~\ref{om2}, for $b\neq0$
	\[
	\begin{aligned}
	\wt(\mathbf c_b)
	&=n-(N_a(b)-1)=p^{m-2}-N_a(b)\\
	&=(p-1)p^{m-3}
	-(p-1)\varepsilon_mp^{(m-5)/2}S_a(b).
	\end{aligned}
	\]
	Clearly, $|s_a(\beta)|\le p+1$. For $m\geq5$,
	\[
	\wt(\mathbf c_b)\ge
	(p-1)p^{(m-5)/2}
	\bigl(p^{(m-1)/2}-(p+1)\bigr)>0
	\]holds for all $b\neq0$. Therefore the map
	\[\phi:\quad
	\BF_q\longrightarrow\mathbb F_p^{\#D_a},\qquad
	c\longmapsto(\Tr(cx))_{x\in D_a}
	\]is injective. Then $\dim(\C_{D_a})=m$.


	Finally, the reachable set required to determine the minimum distance is given by Eq.~\eqref{mo10}.  Note that $R_e(\beta)\geq0$ for $e\geq1$, if $e=1$, then for $\beta\neq0$,
	\[
	R_1(\beta)=1+\eta_K(\beta)
	\]
	equals $2$ on squares and $0$ on non-squares. If $e\geq3$, then for $\beta\neq0$,
	\[
	  R_e (\beta)
	=q_K^{e-1}+\delta_e\eta_K(\beta)q_K^{(e-1)/2}
	\ge   q_K ^{e-1}  -q_K ^{(e-1)/2}>0.
	\]In addition, $R_e(0)=q_K^{e-1}\geq1$. Hence all $\beta\in K$ are attained by some non-zero $b$, i.e., $\mathcal{B}_o$ is precisely the set of $\beta$-values corresponding to non-zero parameters $b$, and taking the maximum reachable $s_a(\beta)$ in ~\eqref{mo8} yields ~\eqref{mo13}.
\end{proof}

\begin{remark}
	If  $d=m, e=1$, then ~\eqref{mo10} degenerates into
	\[
	R_1(0)=1,\qquad R_1(\beta)=1+\eta_{q}(\beta),
	\]that is, only non-zero squares $\beta=b^2$ are reachable, each having two preimages. Therefore, in this special case, the maximization in ~\eqref{mo13} must also be carried out solely over non-zero squares.
\end{remark}
\begin{example}
	Take $p=3$, $m=5$, and let $a\in\BF_{3^5}$ be a root of the irreducible
	polynomial $X^5+2X+1\in\BF_3[X]$. Then $K=\BF_3(a)=\BF_{3^5}$,
	$d=5$, and $e=1$. An exhaustive computation gives
	\[
	n=26,\qquad \dim\C_{D_a}=5,\qquad d_{\min}=10,
	\]
	and
	\[
	W_{\C_{D_a}}(z)=1+6z^{10}+14z^{12}+40z^{14}
	+46z^{16}+56z^{18}+50z^{20}
	+20z^{22}+10z^{24}.
	\]
\end{example}
\begin{theorem}
	Suppose $m=3$, the following holds.
	\begin{enumerate}[1]
		\item $D_a$ is exactly a free $\BF_p^*$-multiplicative orbit, i.e., for any $r\in D_a$
		\[D_a=\{t\cdot r:t\in\BF_p^*\}.\]
		\item The kernel of the map $\phi$ is
		\begin{equation}\label{ker}\ker\phi=\{b\in \BF_{p^3}:\Tr(br)=0\},\end{equation}which is a $2$-dimensional $\BF_p$-hyperplane of size $p^2$.
		\item $\C_{D_a}$ is a $[p-1,1,p-1]_p$ code, whose weight enumerator is
		\begin{equation}\label{mo14}
		W_{\C_{D_a}}(z)=1+(p-1)z^{p-1}.
		\end{equation}
		Counting by parameter $b$, the frequency of zero weight is $p^2$, the frequency of weight $p-1$ is $p^2(p-1)$, and  the minimum distance is $p-1$.
	\end{enumerate}
\end{theorem}
\begin{proof}
	By Eq.~\eqref{mo2}, $\#D_a=p-1$. The scalar multiplication action is free and preserves the two trace constraints, so there is only one orbit. Thus,
	\[
	\mathbf c_b=(t\Tr(br))_{t\in\BF_p^*}.
	\]
	This immediately yields ~\eqref{ker} and ~\eqref{mo14}.
	
	Although Eq.~\eqref{mo6}   contains  $p^{-1}$ when $m=3$, there is no integrality  problem. In this case,
	\[
	S_a(b)=
	\begin{cases}
	p\varepsilon_3=p\eta_p(-1),& 0\neq b\in\ker(\phi),\\
	0,& b\notin\ker(\phi).
	\end{cases}
	\]
	When   $p\equiv1$ or $3\pmod4$, the value inside the kernel is $p$ or $-p$, and substituting into the above weight formula yields weights 0 and $p-1$, respectively.
\end{proof}

\subsection{Illustrative codes}
In this subsection, we shall present two classes of illustrative codes corresponding to the cases $2\mid e$ and $d=3$.
\smallskip

\noindent\textbf{A. The case for even $e$}
\smallskip

When $2\mid e$, we get $d\mid m/2$, and $\eta_q(y)=\eta_K(y)^e=1$ for any $y\in K^*$. For $\beta\in K$, by Eq.~\eqref{Za},
\begin{equation}\label{Za1}
Z_a(\beta)
=\mathrm{1}_{\{\Tr_{K/\BF_p}(\beta)=0\}}+\sum_{t\in\BF_p}
\mathrm{1}_{\{\Tr_{K/\BF_p}(\beta/(a-t))=0\}}.
\end{equation}
Define
\[
\Omega=\left\{\frac1a,\ 1,\ \frac1{1+\lambda a}:\lambda\in\BF_p^*\right\}.
\]
 This set represents the same $p+1$ trace-zero conditions as
	$\{1\}\cup\{(a-t)^{-1}:t\in\BF_p\}$ up to multiplication by elements of $\BF_p^*$.
 Note that all elements in set $\Omega$ are distinct from each other and $\#\Omega=p+1$. For $b\in \BF_q$, set
\[
z(b)=\#\{\omega\in\Omega:\Tr(\omega b^2)=0\}.
\]
We claim that $Z_a(\beta)=z(b)$  for $\beta=\Tr_{\BF_q/K}(b^2)$ . From Lemma~\ref{bb}, we know
\[
\begin{aligned}
R_e(0)&\geq q_K^{e-1}-(q_K-1)q_K^{e/2-1}>0,\\
R_e(\beta)&\ge q_K^{e-1}-q_K^{e/2-1}>0\qquad(\beta\ne0).
\end{aligned}
\]
This means that for any $\beta\in K$, there must exist $b\in\BF_q$ such that
\[\beta=\Tr_{\BF_q/K}(b^2).\]
Replacing $\beta$ with $\Tr_{\BF_q/K}(b^2)$ in ~\eqref{Za1} gives $Z_a(\beta)=z(b)$ , since multiplying a trace coefficient by a nonzero element of $\BF_p$ does not change whether its trace is zero .

By Lemmas~\ref{n0} and ~\ref{three},
\[
\begin{aligned}
h(a)&=p+1,\\
n&=p^{m-2}-\delta(p^2-1)p^{m/2-2}-1,\\
N_a(b)&=p^{m-3}-\delta(p-1)p^{m/2-2}z(b).
\end{aligned}
\]
We further obtain the weight
\[\wt(\mathbf c_b)=(p-1)(p^{m-3}-\delta p^{m/2-2}(p+1-z(b))) \qquad(b\ne0) .\]
The following lemma is very useful for determining the frequency.

\begin{lemma}
	For any $S\subseteq\Omega$, set $k=\# S$ and
	\[
	E_k=\#\{b\in \BF_q:\Tr(\omega b^2)=0\text{ for all }\omega\in S\}.
	\]
	Then
	\[
	\dim_{\BF_p}\langle S\rangle=\min\{d,k\}.
	\]
	Furthermore, the number $E_k$ depends only on $k$ and
\begin{equation*}
	E_k=p^{m-r_k}-\delta(p^{r_k}-1)p^{m/2-r_k},
	\end{equation*}	where $r_k=\min\{d,k\}$.
\end{lemma}
\begin{proof}
	The set $\Omega$ is obtained, up to nonzero $\BF_p$-scalars, from
	 \[\Omega'=\{\frac{1}{a},1,\frac{1}{\lambda+a}:\ \lambda\in\BF_p^*\}.\]
	
	Take $k$ elements from the set $\Omega$ (contain $1$), after common denominator, a non-zero $\BF_p$-polynomial with degree at most $k-1<d$ is obtained to be zero at $a$, which contradicts the fact that the minimal polynomial degree of $a$ is $d$. Therefore, any set of $k\leq d$ elements is linearly independent over $\BF_p$.
	\vspace{6pt}
	
	If $k>d$, all these $k$ elements belong to $K=\BF_p(a)$, and any $d$ of them are linearly independent, so the rank of $S$ is $d$.	Hence,
	\[
	\dim_{\BF_p}\langle S\rangle=\min\{d,k\}.
	\]
	Note that,
\begin{equation*}
	\sum_{b\in\BF_q} \chi _q(ub^2 )= \eta _q (  u )  G_q=-\delta p^{m/2} \quad(u\in K^*) .
	\end{equation*}
	For every $u\in K^*$,
	\[u^{(p^m-1)/2}=(u^{p^{m/2}-1})^{(p^{m/2}+1)/2}=1.\]
	Therefore, any non-zero $\BF_p$-linear combination of $\Omega$ has a quadratic characteristic equal to $1$ in $\BF_q$.
	
	Set $S=\{\omega_1,\omega_2,\ldots,\omega_k\}$.	From the orthogonality of additive characteristic, we have
	\[
	E_k=\frac1{p^k}\sum_{\mathbf y\in\BF_p^k}
	\sum  _{b\in \BF_q}
	 \chi _q \left(\Bigl(\sum_{j=1}^k y_j\omega_j\Bigr)b^2\right).
	\]
	Note that the rank of the following coefficient mapping
	\[\BF_p^k\longrightarrow \BF_p(a),\qquad \mathbf y\longmapsto\sum_jy_j\omega_j\]
	is $r_k$, thus its kernel has $p^{k-r_k}$ elements. Hence
	\[
	\begin{aligned}
	E_k
	&=p^{-k}\left(p^{k-r_k}p^m
	+(p^k-p^{k-r_k})(-\delta p^{m/2})\right)\\
	&=p^{m-r_k}-\delta(p^{r_k}-1)p^{m/2-r_k}.
	\qedhere
	\end{aligned}
	\]
\end{proof}
Take $L=p+1$ and
\begin{equation}\label{eq5}
M_z=\binom{L}{z}
\sum_{i=0}^{L-z}(-1)^i
\binom{L-z}{i}E_{z+i}
-\mathrm{1}_{\{z=L\}},
\end{equation}where $0\le z\le L$.
\begin{theorem}\label{th:weight5}
	 Suppose $2\mid e$ and $m\ge6$.
	 The code $\C_{D_a}$ has parameters $\left[
	p^{m-2}-\delta(p^2-1)p^{m/2-2}-1,\ m
	\right]_p$ and
\begin{equation*}
	\begin{aligned}
	w_z&=(p-1)\bigl[p^{m-3}-\delta p^{m/2-2}(p+1-z)\bigr],\\
	A_{w_z}&=M_z,\qquad 0\le z\le p+1.
	\end{aligned}
	\end{equation*}
	Furthermore, its minimum distance is
	\begin{equation}\label{eq7}
	d_{\min}=
	\begin{cases}
	(p-1)p^{m-3},&\delta=-1,\\[2mm]
	(p-1)\bigl(p^{m-3}-p^{m/2-1}\bigr),
	&\delta=1,\ d=2,\\[2mm]
	(p-1)\bigl(p^{m-3}-(p+1)p^{m/2-2}\bigr),
	&\delta=1,\ d\ge3.
	\end{cases}
	\end{equation}
\end{theorem}
\begin{proof}
	First, fixing exactly $z$ trace conditions in $\Omega$ to be zero, and then apply the inclusion-exclusion principle to the remaining $L-z$ conditions. From the definition of $E_k$, the common zeros of any $k$ conditions are $E_k$. Then the number of the remaining $L-z$ elements do not satisfy the condition that trace equals $0$ is
	\[\sum_{i=0}^{L-z}(-1)^i
	\binom{L-z}{i}E_{z+i}.\]
	
	Note that multiplying by $\binom{L}{z}$ corresponds to selecting the set of zero positions and the element $b=0$ exactly belongs to the class $z=L$, so we subtract 1 from $M_z$(when $z=L$). This yields Eq.~\eqref{eq5}.
	
	This also proves that each $M_z$ is a   non-negative  integer, it is the number of non-zero $c$ satisfying $z(c)=z$. All non-zero $c$ are partitioned into these sets without omission or duplication, thus obtaining
\begin{equation*}
	\sum_{z=0}^{p+1}M_z=p^m-1.
	\end{equation*}

	Now we determine the minimum distance.
	
	If $\delta=-1$, the weight strictly decreases with $z$, and \[
	M_L=E_L-1>0,
	\]
	so the minimum weight is $w_L=(p-1)p^{m-3}$.
	
	If $\delta=1$ and $d=2$, the $p+1$ elements of $\Omega$ give all $p+1$ projection points of the two-dimensional space, every non-zero linear functional vanishes at exactly one point. Therefore $M_0=0$, $M_1>0$, and the minimum distance is $w_1$.
	
	The following completely proves that $M_0>0$ when   $\delta=1,d\ge3$ . Let $f(X)\in\BF_p[X]$ be the minimal polynomial of $a$, so $\deg(f)=d$. For each $t\in\BF_p$, define\begin{equation*}
	g_t(X)=\frac{f(X)-f(t)}{X-t}.
\end{equation*}
	Since $f(a)=0$ and $f(t)\neq0$, we have
	\begin{equation}\label{eq13}
	\frac{1}{a-t}=-\frac{g_t(a)}{f(t)}.
	\end{equation}
	On the other hand, except for the element $1$, the elements in $\Omega$ multiplied by a non-zero $\BF_p$-scaling result exactly in
	\begin{equation}\label{eq14}
	\frac1{a-t},\qquad t\in\mathbb F_p.
	\end{equation}
	Consider the mapping
	\begin{equation}\label{eq15}
	\begin{aligned}
	\Phi:K^\vee&\longrightarrow
	\{h(T)\in\BF_p[T]:\deg h\le d-1\},\\
	\ell&\longmapsto \Phi(\ell)(T)=\ell(g_T(a)),
	\end{aligned}
	\end{equation}
	where   $K^\vee=\operatorname{Hom}_{\BF_p}(K,\BF_p)$ . Set
	\[f(X)=\sum_{i=0}^{d}f_iX^i,\qquad f_d=1.\]
	Then
	\begin{equation}\label{eq16}
	g_T(a)=\sum_{i=1}^{d}f_i
	\sum_{j=0}^{i-1}a^{i-1-j}T^j.
	\end{equation}
	In Eq.~\eqref{eq16}, the coefficients of $T^{d-1},T^{d-2},\ldots,T^0$ successively have degrees regarding $a$ of \[
	0,1,\ldots,d-1,
	\]
	and the leading coefficients are all $1$, these coefficients thus constitute an $\BF_p$-basis of $K$. Therefore, Eq.~\eqref{eq15} is an isomorphism between two $d$-dimensional $\BF_p$-spaces.
	
	Take any monic irreducible polynomial of degree $d-1$ \[
	h(T)\in\mathbb F_p[T].
	\]Since $d-1\geq2$,  it has no roots in $\BF_p$. By the surjectivity of $\Phi$, there exists a unique   $\ell\in K^\vee$  such that \begin{equation*}
	\ell(g_t(a))=h(t)\qquad(t\in\BF_p).
\end{equation*}
	Furthermore, because the coefficient of $T^{d-1}$ in $g_T(a)$ is $1$, comparing the highest degree terms on both sides of Eq.~\eqref{eq15} yields \begin{equation}\label{eq18}
	\ell(1)=1.
	\end{equation}
	From Eqs.~\eqref{eq14}, ~\eqref{eq18} and $h(t)f(t)\ne0$, we have
	\[
	\ell\!\left(\frac1{a-t}\right)
	=-\frac{h(t)}{f(t)}\ne0
	\qquad(t\in\mathbb F_p).
	\]
	Combining ~\eqref{eq13} and ~\eqref{eq18}, we know
\begin{equation*}
	\ell(\omega)\ne0\qquad\text{for all }\omega\in\Omega.
\end{equation*}
	It remains to be proved that this abstract linear functional can indeed be realized by some $c^2$. Take any $\BF_p$-basis $\beta_1,\beta_2,\ldots,\beta_d$ of $K$ and let
	\[
	b_i=\ell(\beta_i).
	\]
	Since $\ell\neq0$, the vector $\mathbf b=(b_1,\ldots,b_d)$ is non-zero. Denote
	\[
	R(\mathbf b)=
	\#\{c\in \BF_q:\Tr(\beta_i c^2)=b_i,\ 1\le i\le d\}.
	\]
	Expanding using the additive character relation:
\begin{equation*}
	R(\mathbf b)
	=\frac1{p^d}\sum_{\mathbf y\in\BF_p^d}
	\exp\left(-\frac{2\pi i}{p}\mathbf y\cdot\mathbf b\right)
	\sum_{c\in \BF_q}
	\chi_q\left(\Bigl(\sum_{i=1}^dy_i\beta_i\Bigr)c^2\right).
\end{equation*}
	When $\mathbf y=0$, the inner sum is $p^m$. When $\mathbf y\ne0$,
	\[
	A_{\mathbf y}=\sum_i y_i\beta_i\in K^*\subseteq \BF_{p^{m/2}}^*
	\]is a square element in $\BF_q$, so the corresponding inner sum equals $-\delta p^{m/2}$. Since
	$\mathbf b\ne0$,
	\[
	\sum_{\mathbf y\ne0}
	\exp\left(-\frac{2\pi i}{p}\mathbf y\cdot\mathbf b\right)=-1.
	\]
	Therefore,
\begin{equation*}
	R(\mathbf b)
	=p^{-d}\bigl(p^m+  \delta  p^{m/2}\bigr).
\end{equation*}
	In particular, in the case of $\delta=1$,
	\[
	R(\mathbf b)=p^{-d}(p^m+p^{m/2})>0.
	\]
	Therefore, there exists $c\in \BF_q$ satisfying
	\[
	\Tr(\beta_i c^2)=\ell(\beta_i)\qquad(1\le i\le d).
	\]By linearity, this means
	\[\Tr(\omega c^2)=\ell(\omega)\ne0
	\qquad\text{for all }\omega\in\Omega.
	\]
	This $c$ is non-zero, and $z(c)=0$. So $M_0>0$.
\end{proof}
\begin{example}
Take $p=5$, $m=6$, and let $a\in\BF_{5^3}\subset\BF_{5^6}$ be a root of
$T^3+2T+1\in\BF_5[T]$. Then $d=3$, $e=2$, and $\delta=1$.
An exhaustive computation gives
\[
n=504,\qquad \dim\C_{D_a}=6,\qquad d_{\min}=380,
\]
and
\[
W_{\C_{D_a}}(z)
=1+5040z^{380}+3024z^{400}+7560z^{420}.
\]
\end{example}
\smallskip

\noindent\textbf{B. The case for $d=3$}
\smallskip

Let \[
f_a(T)=T^3+AT^2+BT+C\in\BF_p[T]
\] be a minimal polynomial of $a$, and $K=\BF_p(a)=\BF_{p^3}$.
\begin{lemma}\label{spec1}
	Write
	\[
	\beta=x+ya+za^2,\qquad x,y,z\in\BF_p.
	\]
	For any $c=u+va\ne0$, we then have
	\begin{equation}\label{cas1}
	N_{K/\BF_p}(c)
	=u^3-Au^2v+Buv^2-Cv^3
	\end{equation}
	and
	\begin{equation}\label{cas2}N_{K/\BF_p}(c)\Tr_{K/\BF_p}(\frac{\beta}{c})=Q_\beta(u,v)=a_\beta u^2+b_\beta uv+c_\beta v^2,\end{equation}
	where $a_\beta=3x-Ay+(A^2-2B)z$, $b_\beta=-2Ax+2By+(3C-AB)z$ and $c_\beta=Bx-3Cy+ACz$.
	
	In particular,
\begin{equation*}
	s_a(\beta)=
	\sum_{[u:v]\in\mathbb P^1(\BF_p)}
	\eta_p\bigl(Q_\beta(u,v)\bigr).
\end{equation*}
\end{lemma}
\begin{proof}
	We have
	\[N_{K/\BF_p}(c)=N_{K/\BF_p}(u+va)=(u+va)(u+va)^p(u+va)^{p^2}.\]
	By Vieta's formulas,
	\[a +a^p+a^{p^2}=-A,\ aa^p+aa^{p^2}+a^pa^{p^2}=B,\ a\cdot a^p\cdot a^{p^2}=-C.\]
	Hence, Eq.~\eqref{cas1} can be obtained. Note that
	\[
	(u+va)^{-1}
	=
	\frac{
		u^2+uv(-A-a)+v^2(a^2+Aa+B)
	}{N_{K/\BF_p}(u+va)}.
	\]
	Multiplying by $\beta$ and taking the trace yields ~\eqref{cas2}. Using Newton's identities
	\[
	\begin{aligned}
	\Tr_{K/\BF_p}(a)&=-A,\qquad \Tr_{K/\BF_p}(a^2)=A^2-2B,\\
	\Tr_{K/\BF_p}(a^3)&=-A^3+3AB-3C,\\
	\Tr_{K/\BF_p}(a^4)&=A^4-4A^2B+2B^2+4AC.
	\end{aligned}
	\]gives $a_\beta, b_\beta, c_\beta$. Since $\eta_K(c)=\eta_p(N_{K/\BF_p}(c))$, by Eq.~\eqref{mo8}, we have
	\[
	\begin{aligned}
	s_a(\beta)
	&=\sum_{[u:v]\in\mathbb P^1(\BF_p)}
	\eta_p(N_{K/\BF_p}(c))
	\eta_p\!\left(\Tr_{K/\BF_p}\frac{\beta}{c}\right)\\
	&=\sum_{[u:v]\in\mathbb P^1(\BF_p)}\eta_p(Q_\beta(u,v)).
	\qedhere
	\end{aligned}
	\]
\end{proof}

Set
\[
\Psi_a:\ K\longrightarrow\Sym^2(\BF_p^2)^*,
\qquad \beta\longmapsto Q_\beta(u,v).
\]
Note that $1, -A-a, a^2+Aa+B\in K$ form a basis of $K/\BF_p$ and
\[
\begin{aligned}
a_\beta&=\Tr_{K/\BF_p}(\beta),\\
b_\beta&=\Tr_{K/\BF_p}(\beta(-A-a)),\\
c_\beta&=\Tr_{K/\BF_p}(\beta(a^2+Aa+B)).
\end{aligned}
\]
Since the trace map is non-degenerate, hence the map $\Psi_a$ is an isomorphism.
\begin{proposition}
	Set \[
	\widehat{s_a}(\gamma)
	=\sum_{\beta\in K}s_a(\beta)
	\chi_p\!\left(\Tr_{K/\BF_p}(\gamma\beta)\right).
	\]
	If $U=\BF_p+\BF_pa$, then
\begin{equation*}
	\widehat{s_a}(\gamma)=
	\begin{cases}
	p^2G_p\,\eta_K(\gamma),
	&\gamma\ne0,\ \gamma^{-1}\in U,\\
	0,& \text{otherwise}.
	\end{cases}
\end{equation*}
\end{proposition}
\begin{proof}
	  Summing over the projective classes $[c]\in\mathbb P(U)$ represented by nonzero elements $c=u+va$ , we get
	\[
	\widehat{s_a}(\gamma)
	=\sum_{\beta\in K}\sum_{[c]\in \mathbb P(U)}
	\eta_p(N_{K/\BF_p}(c))\cdot\eta_p\!\left(\Tr_{K/\BF_p}\frac{\beta}{c}\right)
	\chi_p\!\left(\Tr_{K/\BF_p}(\gamma\beta)\right).
	\]
	Note that $\widehat{s_a}(\gamma)=0$ if $\gamma=0$. Assume that $\gamma\neq0$ and
	$\gamma=k/c$, then
	\[
	\sum_{\beta\in K}
	\eta_{p}(
	\Tr_{K/\BF_p}
	(\frac{\beta}{c}))\chi_p\!\left(k\Tr_{K/\BF_p}(\frac{\beta}{c})\right)\]is nonzero only when $k\in\BF_p^*$. Take $\beta'=\beta/c$. Then the sum becomes
	\begin{equation}\label{sab}
	\sum_{t\in\BF_p}\eta_p(t)
	\sum_{\substack{\beta'\in K\\ \Tr_{K/\BF_p}(\beta')=t}}
	\chi_p(\Tr_{K/\BF_p}(k\beta')).
	\end{equation}
	Now we discuss $k$ in three cases:
	\begin{itemize}
		\item If $k=0$, then the sum of Eq.~\eqref{sab} is $0$.
		\item If $k\in\BF_p^*$, then $\gamma^{-1}\in U$
\begin{equation*}
		\begin{aligned}
		&\sum_{t\in\BF_p}\sum_{\substack{\beta'\in K\\ \Tr_{K/\BF_p}(\beta')=t}}
		\eta_p(t)\chi_p(k\Tr_{K/\BF_p}(\beta'))\\
		&\quad=p^2\sum_{t\in\BF_p}\eta_p(t)\chi_p(kt)
		=p^2\eta_p(k)G_p.
		\end{aligned}
\end{equation*}
		Therefore
		\[
		\begin{aligned}
		\widehat{s_a}(\gamma)
		&=\eta_p(N_{K/\BF_p}(c))\eta_p(k)p^2G_p\\
		&=p^2G_p\eta_p(kN_{K/\BF_p}(c))\\
		&=p^2G_p\eta_K(\gamma).
		\end{aligned}
		\]
		\item If $k\notin \BF_p$, then the mappings $\beta'\mapsto\Tr_{K/\BF_p}(\beta')$ and $\beta'\mapsto\Tr_{K/\BF_p}(k\beta')$ are two linearly independent linear functional. For such a bilinear pairing, we can further make a variable substitution: for any fixed $t$, we let \[V_t:=\{\beta'\in K: \Tr_{K/\BF_p}(\beta')=t\}.\]
	Then $\Tr_{K/\BF_p}(k\beta')$ is surjective on $V_t$, i.e., it can take all values in $\BF_p$ uniformly. Therefore
		\[
		\sum_{\substack{\beta'\in K\\ \Tr_{K/\BF_p}(\beta')=t}}
		\chi_p(\Tr_{K/\BF_p}(k\beta'))
		=p\sum_{s\in\BF_p}\chi_p(s)=0.
		\]
		So $\widehat{s_a}(\gamma)=0$.
	\end{itemize}
	Therefore, the proof is finished.
\end{proof}
We proceed to enumerate all values of $s_a(\beta)$ and their respective frequencies.
\begin{theorem}\label{valsa}
	For a parameter $a$ defining a cubic extension $\BF_p(a)/\BF_p$, we have the distribution of $s_a(\beta)$ given in Table I.
	\begin{table}[htbp]
		\centering
		\caption{The distribution of $s_a(\beta)$}
		\label{tab:joint-dist}
		\begin{tabular}{ccc}
			\toprule
			$s_a(\beta)$ & $\eta_K(\beta)$ & $\#\{\beta\}$ \\
			\midrule
			$-p$ & $-\theta$ & $(p^2-1)/2$ \\[2pt]
			$0$  & $-1$    & $(p^3-p^2)/2$ \\
			$0$  & $0$     & $1$ \\
			$0$  & $1$     & $(p^3-p^2)/2$ \\[2pt]
			$p$  & $\theta$  & $(p^2-1)/2$ \\
			\bottomrule
		\end{tabular}
	\end{table}

	Therefore, this distribution depends only on $p$, and not on the minimal polynomial of $a$, the affine linear orbit, the fractional linear orbit, or the Frobenius conjugate orbit.
	
	More precisely, the only standard arithmetic invariant entering the proof is
\begin{equation*}
	\eta_K(f_a'(a))
	=\eta_p\!\left(-\disc(f_a)\right)
	=\theta.
\end{equation*}
	In particular, the discriminant of an irreducible cubic polynomial lies in $\BF_p^*$ and has fixed quadratic character, so the distribution depends only on $p$.
\end{theorem}
\begin{proof}
	We have
	\[
	\begin{aligned}
	\Sigma(Q_\beta)&:=\sum_{[u:v]\in\mathbb P^1(\BF_p)}\eta_p(Q_\beta(u,v))\\
	&=\sum_{[u:v]\in\mathbb P^1(\BF_p)}
	\eta_p(a_\beta u^2+b_\beta uv+c_\beta v^2).
	\end{aligned}
	\]
	The projective line $\mathbb P^1(\BF_p)$ contains exactly $p+1$ points, we decompose the sum by splitting the projective points into $[1:0]$ and $[t:1]$, $t\in\BF_p$. Then
	\begin{align*}
	\Sigma(Q_\beta)&=\sum_{t\in\BF_p}\eta_p(Q_\beta(t,1))+\eta_p(Q_\beta(1,0))\\
	&=\sum_{t\in\BF_p}\eta_p(a_\beta t^2+b_\beta t+c_\beta)
	+\eta_p(a_\beta)\\
	&=\sum_{t\in\BF_p}\eta_p(a_\beta t^2+b_\beta t+c_\beta)+\eta_p(a_\beta).
	\end{align*}
	
	Set $a_\beta\neq0$. The quadratic form $Q_\beta(u,v)$ has rank 2 if and only if its discriminant $\Delta=b_\beta^2-4a_\beta c_\beta\neq0$, then
	\[
	\sum_{t\in\BF_p}\eta_p(a_\beta t^2+b_\beta t+c_\beta)
	=\eta_p(a_\beta)\sum_{y\in\BF_p}\eta_p(y^2+D),
	\]
	where $y=t+b_\beta/(2a_\beta)$ and $D=(4a_\beta c_\beta-b_\beta^2)/(4a_\beta^2)$. Note that
	\[
	\begin{aligned}
	&\sum_{y\in\BF_p}(1+\eta_p(y^2+D))\\
	&\quad=\#\{(x,y)\in\BF_p^2:x^2-y^2=D\}=p-1.
	\end{aligned}
	\]
	Then   $\sum_{y\in\BF_p}\eta_p(y^2+D)=-1$  and
	\[\Sigma(Q_\beta)=-\eta_p(a_\beta)+\eta_p(a_\beta)=0.\]
	
	If $a_\beta=0$, $Q_\beta(u,v)$ has rank 2 if and only if $b_\beta\neq0$. Then
	\[
	\begin{aligned}
	\Sigma(Q_\beta)&=\sum_{t\in\BF_p}\eta_p(Q_\beta(1,t))+\eta_p(Q_\beta(0,1))\\
	&=\sum_{t\in\BF_p}\eta_p(b_\beta t+c_\beta t^2)+\eta_p(c_\beta).
	\end{aligned}
	\]
	By the above discussion, we also have $\Sigma(Q_\beta)=0$.
	
	If the quadratic form $Q_\beta(u,v)$ has rank $1$, then there exist a constant $\lambda$ and a linear form $L(u,v)=ru+sv\neq0$ such that
	\[Q_\beta(u,v)=\lambda (L(u,v))^2.\]
	Note that the only zero of $L(u,v)$ in $\mathbb P^1(\BF_p)$ is exactly one point $[u:v]=[s:-r]$. Then $\eta_p(Q_\beta(s,-r)) = \eta_p(0) = 0$. For other nonzero points, we know
	\[\eta_p\left(Q_\beta(u,v)\right) = \eta_p\left(\lambda \cdot L(u,v)^2\right) = \eta_p(\lambda) \cdot 1 = \eta_p(\lambda).\] Therefore,
	\[\Sigma(Q_\beta)=p\eta_p(\lambda).\]
	
	Since the map $\beta\mapsto Q_\beta$ is an $\BF_p$-isomorphism, then
	\[\#\{\beta\in K:\ \rank(Q_\beta)=1\}=(p+1)(p-1)=p^2-1, \]
	and the remaining $p^3-p^2$ nonzero quadratic forms all have rank $2$.
	
	We next determine the value of $\eta_K(\beta)$ when $\rank(Q_\beta)=1$. Take $\Gal(K/\BF_p)=\{\mathrm{id},\sigma,\sigma^2\}$, where $\sigma(x)=x^p$ for any $x\in K$. Then $Q_\beta(u,v)$ can be rewritten as
	\begin{equation*}
	Q_\beta(u,v)=\sum_{g\in\Gal(K/\BF_p)}g(\beta^{p^2}(u+av)\sigma(u+av)).
	\end{equation*}
	Take $v=1$ and $u=t$, we have
	\[
	\begin{aligned}
	P_\beta(t)&:=Q_\beta(t,1)\\
	&=\sum_{g\in\Gal(K/\BF_p)}g(\beta^{p^2}(t+a)\sigma(t+a))\\
	&=a_\beta t^2+b_\beta t+c_\beta.
	\end{aligned}
	\]
	Then
	\[
	\begin{aligned}
	P_\beta(-a)
	&=\beta\prod_{\mathrm{id}\neq g\in\Gal}(-a+g(a))\\
	&=\beta\prod_{\mathrm{id}\neq g\in\Gal}(a-g(a))
	=f'_a(a)\beta.
	\end{aligned}
	\]
	
	Since $Q_\beta(u,v)=\lambda L(u,v)^2=\lambda(ru+sv)^2$, we get
	\[
	\beta=\lambda\,\frac{(s-ra)^2}{f_a'(a)}.
	\]
	Hence, \begin{equation}\label{disc}
	\eta_K(\beta)
	=\eta_K(\lambda)\eta_K(f_a'(a))=\eta_p(\lambda)\eta_K(f_a'(a)).
	\end{equation}
	Note that $f_a(x)=\prod\limits_{i=0}^{2}(x-a^{\sigma^i})$, then its discriminant is
	\[\disc(f_a)=\prod_{0\leq i<j\leq 2}(a^{\sigma^i}-a^{\sigma^j}) ^2 =\det(M_{\rm van})^2,\]
	where
	\[M_{\rm van}=\begin{pmatrix}
	1 & a & a^2 \\
	1 & a^{p} & a^{2p} \\
	1  & a^{p^2} & a^{2p^2}
	\end{pmatrix}.\]
	Note that Frobenius cyclically permutes the rows of $M_{\rm van}$, and this permutation has sign $(-1)^{3-1}=1$. Hence $\det(M_{\rm van})$ is fixed by $\sigma$, and
	\[
	\disc(f_a)=\sigma(\disc(f_a))
	\implies \disc(f_a)\in\BF_p
	\]Then \[\eta_p(\disc(f_a))=1.\]
	Since $\disc(f_a)=-N_{K/\BF_p}(f_a'(a))$ and $s_a(\beta)=p\eta_p(\lambda)$, by ~\eqref{disc},
	\[
	\eta_K(\beta)=\theta\,\frac{s_a(\beta)}p.
	\]
	
	Since $\sum_{x\in K}\eta_K(x)=0$, the total contribution of all rank-1 quadratic forms to the  character  sum is zero  . Therefore the  rank-2 forms   are split evenly between square and nonsquare values of $\eta_K(\beta)$ .
	
	Note that
	\[\#\operatorname{PGL}_2(\mathbb F_p)=\frac{\#\GL_2(\BF_p)}{\#\BF_p^*}=p(p^2-1)=p^3-p
	=\#(K\setminus\mathbb F_p),\]
	$\operatorname{PGL}_2(\mathbb F_p)$ acts on the projective line $\mathbb P^1(\overline{\BF}_p)$ via fractional linear transformations, i.e.,
	\[g(z)=\frac{rz+s}{tz+u},\quad g=\begin{pmatrix}
	r & s \\
	t & u
	\end{pmatrix}\in \GL_2(\BF_p)\]
	Suppose $g(a)=a$,
	\[
	\begin{aligned}
	\frac{ra+s}{ta+u}=a
	&\implies ra+s=a(ta+u)\\
	&\implies ta^2+(u-r)a-s=0.
	\end{aligned}
	\]
	Since the minimal polynomial of $a$ has degree 3 over $\BF_p$, there is no nonzero degree$\leq2$ polynomial vanishing at $a$.  Thus all coefficients must be zero
	\[t=0,\quad u-r=0,\quad s=0,\]this means $g=r\cdot \mathrm{Id}_2$, where $\mathrm{Id}_2$ is the identity matrix of $\GL_2(\mathbb F_p)$ and $[\mathrm{Id}_2]$ is the identity element of $\operatorname{PGL}_2(\mathbb F_p)$. The stabilizer of every $a\in K\setminus \BF_p$ is trivial, orbit-stabilizer theorem gives that the size of any orbit is equal to $\#\operatorname{PGL}_2(\mathbb F_p)=\#(K\setminus\mathbb F_p)=p^3-p$, this means that there is only one orbit, so the action is transitive (and free).
	
	If we restrict to the affine subgroup
	$\operatorname{Aff}_1(\BF_p)=\left\{\begin{pmatrix}
	r & s \\
	0 & 1
	\end{pmatrix}\right\}$ (size $p(p-1)$),
	or  the Frobenius cyclic subgroup generated by $x\mapsto x^p$ (size $3$), then orbits of $a$ become smaller
	\begin{itemize}
		\item Frobenius orbit of $a$ $\{a,a^p,a^{p^2}\}$,
		\item Affine orbit size divides $p(p-1)$.
	\end{itemize}
	
	In Table~\ref{tab:joint-dist}, half quadratic form with $\rank(Q_\beta)=2$ are square, half nonsquare. This holds for each individual parameter $a$ regardless of its orbit. Therefore the count is local to the element, not dependent on its symmetry group.
\end{proof}

By Theorems~\ref{th:weight3} and \ref{valsa}, we have
\begin{theorem}\label{thm:d=3}
	Let $m=3e$,  where  $e\geq3$  is odd . Then the code $\C_{D_a}$ has parameter $[p^{m-2}-1,m]_p$, and the weight enumerator is
	\[W_{\C_{D_a}}(z)=1+B_-z^{(p-1)(p^{m-3}- p^{(m-3)/2})}+B_0z^{(p-1)p^{m-3}}+B_+z^{(p-1)(p^{m-3}+ p^{(m-3)/2})},\]
	where
	\begin{align*}
		B_-=\frac{p^2-1}{2}(p^{m-3}+p^{(m-3)/2}),\\
	B_0=p^{m-3}-1+(p-1)p^{m-1},\\
	B_+=\frac{p^2-1}{2}(p^{m-3}-p^{(m-3)/2}).
	\end{align*}
	Therefore, its minimum distance is
	
	\[d_{\min}=(p-1)(p^{m-3}-p^{(m-3)/2}).\]
\end{theorem}
\begin{example}
Take $p=3$, $m=9$, and let $a\in\BF_{3^3}\subset\BF_{3^9}$ be a root of
$T^3-T-1\in\BF_3[T]$. Then $d=3$ and $e=3$. An exhaustive computation gives
\[
n=2186,\qquad \dim\C_{D_a}=9,\qquad d_{\min}=1404,
\]
and
\[
W_{\C_{D_a}}(z)
=1+3024z^{1404}+13850z^{1458}+2808z^{1512}.
\]
\end{example}
\section{More properties for $\C_{D_a}$}
In this section, we first investigate the parameters of the punctured code   $\C_{R_a}$. Subsequently, we study the self-orthogonality, divisibility and projectivity of both $\C_{D_a}$ and its punctured code, along with the minimum distance of their dual codes.

Select one representative from each $\BF_p^*$-orbit, and denote this set by
\[R_a=\{r_1,r_2,\ldots,r_s\},\]
	where $s=\#D_a/(p-1)$. Then we have

\begin{theorem}  \label{puncode-da} Let $\C_{D_a}$ be a code over   $\BF_p$  with defining set $D_a$, and let $\C_{R_a}$ be its punctured code. Then the code $\C_{R_a}$ has length $\#D_a/(p-1)$, dimension $\dim(\C_{R_a})=\dim(\C_{D_a})=m-\log_p\mathcal H$ and its weight frequencies satisfy
	\[A_j(\C_{R_a})=A_{(p-1)j}(\C_{D_a}).\]
\end{theorem}
\begin{proof}
	Choosing an orbit $[r_i]=\{\lambda r_i: \lambda\in\BF_p^*\}$, its coordinates are
	\[\Tr(b\lambda r_i)=\lambda\Tr(br_i).\]
	
	If $\Tr(br_i)\neq0$, the code $\C_{D_a}$ has all $p-1$ nonzero coordinates in this orbit, otherwise, all are zero. Hence, the weight of $\C_{D_a}$ is exactly $p-1$ times the weight of $\C_{R_a}$, and the kernel of the map $\phi$ is invariant.
\end{proof}

\subsection{Self-orthogonality}
\begin{lemma}[\cite{2}]\label{lemma:3self}
	Let $p=3$ and $\C$ be a linear code over $\BF_3$. Then $\C$ is self-orthogonal if and only if the weights of all codewords are divisible by three.
\end{lemma}
\begin{theorem}\label{th:sef}
	Let $\C_{D_a}$ be a nonzero code over   $\BF_p$  with defining set $D_a$, and let $\C_{R_a}$ be its punctured code. The following holds.
	\begin{enumerate}[1]
		\item If $2\nmid m$, then $\C_{D_a}$ is self-orthogonal if and only if $p>3$ or $p=3$ and $m\ge 7$.
		\item If $2\mid m$, then $\C_{D_a}$ is self-orthogonal if and only if
		$p>3$ or $p=3$ and $m\geq 6$.
	\end{enumerate}
\end{theorem}
\begin{proof}
	In   $\mathbb P^{m-1}(\BF_p)$ , a projective direction  through the origin can be denoted by
	\[[x]=\{tx:\ t\in\BF_p^*\},\quad x\neq0.\]
	Then $D_a$ satisfies the property that  for every projective point   $[x]\in \mathbb P^{m-1}(\BF_p)$ , either all non-zero multiples $tx$ ($t\in\BF_p^*$) are contained in $D_a$ or none are. Then
	\begin{equation}  \label{uni-da} D_a=\bigcup_{i=1}^{s}\{tr_i:\ t\in\BF_p^*\},\ [r_i]\in  \mathbb P^{m-1}(\BF_p).\end{equation}
	Hence, $n=s(p-1)$.
	
	For $b,c\in\BF_{q}$,   the contribution of one orbit is
	 \[
	 \langle \mathbf c_b,\mathbf c_c\rangle_{[r_i]}
	=
	\sum_{t\in\BF_p^*}
	\Tr(btr_i)\Tr(ctr_i)
	=
	\Tr(br_i)\Tr(cr_i)
	\sum_{t\in\BF_p^*}t^2.
	\]
	 Summing over all orbits gives the full inner product.  For $p\geq5$, \[
	\sum_{t\in\BF_p^*}t^2=0,
	\]hence the code $\C_{D_a}$ is self-orthogonal.
	
	For $p=3$ and $b\in\BF_q^*$, by Theorems~\ref{th:weight1} and \ref{th:weight3}, if $2\mid m\ge 6$, then
	\[\wt(\mathbf c_b)=2\left( 3^{m-3}-\delta3^{m/2-2}(h(  a )-Z_a(\beta(b)))\right).\] If $2\nmid m\ge 7$, then \[\wt(\mathbf c_b)=2\left(3^{m-3}-  \varepsilon _m3^{(m-5)/2}s_a(\beta(b))\right).\]
	Therefore, the ternary code is self-orthogonal by Lemma~\ref{lemma:3self}.
	
	For $p=3$ and $m=3$, the code is $[2,1,2]_3$. A non-zero generator has the form $(s,-s)$ with $s\neq0$, and
	\[
	s^2+(-s)^2=2s^2=-s^2\ne0.
	\]
	Hence the code is not self-orthogonal.
	
It remains to consider $m=5$. Necessarily,	\[\BF_3(a)=\BF_{3^5}.\]
Fix an arbitrary $\BF_3$-basis of $\BF_{3^5}$, and regard 
\[Q_0(x)=\Tr(x^2),\qquad Q_a(x)=\Tr(ax^2)\] as quadratic forms in five variables. For the linear forms $L_b$ and $L_c$ corresponding to $b,c\in\BF_{3^5}$, the Gram bilinear form is
\begin{equation}\label{eq:27}
	\sum_{x\in\BF_3^5}L_b(x)L_c(x)
	(1-Q_0(x)^2)(1-Q_a(x)^2).
\end{equation}

In fact, over $\BF_3$, the sum of a monomial over any variable vanishes whenever the degree of that variable is either zero or odd. Since every term in ~\eqref{eq:27} has total degree at most \[10=2m,\]
the only potentially nonzero contribution arises from
\[L_bL_cQ_0^2Q_a^2,\]
and only from those monomials in which every variable occurs with degree exactly two.

Extend the scalars to an algebraic closure. The multiplication operator $x\mapsto ax$ is separable and self-adjoint with respect to the trace pairing. Hence there exist simultaneous orthogonal coordinates such that
\[
Q_0=\sum_{i=1}^5c_i x_i^2,\qquad
Q_a=\sum_{i=1}^5c_i a_i x_i^2,\]where $c_i\neq0$ and $a_1,a_2,\ldots,a_5$ are the five Galois conjugates of $a$. Set
\[F_1=\sum_{i=1}^{5}a_i,\qquad F_2=\sum_{1\le i<j\le 5}a_ia_j.\]

With respect to the corresponding dual coordinates, the Gram matrix associated with ~\eqref{eq:27} is diagonal. Up to a common nonzero scalar, its $i$-th diagonal entry equals
\[
\left(\prod_{j\ne i}c_j\right)
\sigma_2(a_1,\ldots,\widehat{a_i},\ldots,a_5),
	\]where the hat denotes omission and $\sigma_2$ denotes the second elementary symmetric polynomial. Moreover,
\begin{equation}\label{eq:30}
	\sigma_2(a_1,\ldots,\widehat{a_i},\ldots,a_5)
	=a_i^2-F_1a_i+F_2.
\end{equation}
If Eq.~\eqref{eq:30} is zero, then $a_i$, whose degree over $\BF_3$ is 5, would satisfy a monic quadratic polynomial over $\BF_3$, which is impossible.

Therefore, every diagonal entry is nonzero, the Gram matrix has rank 5, and $\C_{D_a}$ is not self-orthogonal for any $a\in\BF_{3^5}\setminus\BF_3$. This completes the proof.
\end{proof}
\begin{example}
Take $p=3$, $m=6$, and let $a\in\BF_{3^3}\subset\BF_{3^6}$ be a root of
$T^3-T-1\in\BF_3[T]$. An exhaustive computation gives $n=104$ and
\[
W_{\C_{D_a}}(z)
=1+52z^{54}+312z^{66}+208z^{72}+156z^{78}.
\]
Moreover, for a generator matrix $G$ of rank $6$, the program returns
\[
GG^T=0_{6\times6}.
\]
Thus $\C_{D_a}$ is self-orthogonal, equivalently, all its nonzero
weights are divisible by $3$.
\end{example}
We next study the self-orthogonality of $\C_{R_a}$.
\begin{theorem}
Let $e_1,\ldots,e_m$ be an $\BF_p$-basis of $\BF_{p^m}/\BF_p$, and let
\[
g(r)=
\bigl(\Tr(e_1r),\ldots,
\Tr(e_mr)\bigr)^T,
\qquad r\in R_a.
\]
The punctured code $\C_{R_a}$ is self-orthogonal if and only if
\begin{equation}\label{mo17-da}
\sum_{r\in R_a}g(r)g(r)^T=0.
\end{equation}
If the representative $r_i$ of each orbit is rescaled to $t_ir_i$ with $t_i\in\BF_p^*$, then the condition becomes
\begin{equation}\label{mo18}
\sum_i t_i^2g(r_i)g(r_i)^T=0.
\end{equation}
\end{theorem}
\begin{proof}
Set $s=\#R_a=\#D_a/(p-1)$ and $R_a=\{r_1,r_2,\ldots,r_s\}$. The generator matrix of $\C_{R_a}$ is
\[
G_{R_a}=(g(r_1),g(r_2),\ldots,g(r_s)).
\]
Hence its Gram matrix is
\[
G_{R_a}G_{R_a}^T=\sum_{r_i\in R_a}g(r_i)g(r_i)^T.
\]
Therefore, $\C_{R_a}$ is self-orthogonal if and only if \eqref{mo17-da} holds. Rescaling $r_i$ to $t_ir_i$ multiplies the corresponding column by $t_i$, giving \eqref{mo18}.
\end{proof}
\begin{remark}\label{re:r3}
	For $p=3$, $t_i^2=1$ for any $t_i\in\BF_3^*$. Hence the Gram matrix is invariant under the choice of orbit representatives and non-zero coordinate scalings. Therefore, the punctured code is self-orthogonal exactly when the original code is self-orthogonal, which occurs precisely for odd $m\geq7$ or even $m\geq6$ by Theorem~\ref{th:sef}.
\end{remark}
Remark~\ref{re:r3} settles the case $p=3$, where the weights are trivial ($t_i^2=1$) and self-orthogonality reduces to the underlying orbit structure. For $p>3$, however,  we must allow non-trivial scaling. This leads to the notion of a weighted punctured code, where each projective point $P\in X_a$ is assigned a squared weight $c_P\in \mathcal S_p$.

Take \[\mathscr X_a=V(Q_0,Q_a)\subseteq\mathbb P^{m-1},\qquad
X_a=\mathscr X_a(\BF_p)=D_a/\BF_p^*,\qquad \mathcal S_p:=(\BF_p^*)^2,\qquad P=[r_i]\in X_a.\]
\begin{definition}
	For each $P\in X_a$, choose a representative $r_P$ and let $g_P=g(r_P)=(\Tr(e_1r_P),\ldots,\Tr(e_mr_P))^T\in\BF_p^m$. Assign a weight $c_P\in \mathcal S_p$ to each point $P$. Choose $\gamma_P\in\BF_p^*$ with $\gamma_P^2=c_P$ and define the weighted punctured code $\C$ as the row space of the matrix whose $P$-th column is $\gamma_Pg_P$. Its Gram matrix is
\begin{equation*}
	\sum_{P\in X_a}c_Pg_Pg_P^T\in\Sym_m(\BF_p),\qquad c_P\in \mathcal S_p.
\end{equation*}
\end{definition}
\begin{remark}
	For each $P\in X_a$, select $\gamma_P\in\BF_p^*$ such that $c_P=\gamma_P^2$. Then $\C$ can be generated by the vectors $\gamma_P g_P$, and it is monomially equivalent to the code $\C_{R_a}$. 
	
	If $D_a\neq\emptyset$, the parameters of $\C$ are given by
	\[\operatorname{length}\ \C=s,\qquad \dim\C=\dim\C_{D_a},\qquad d_{\min}(\C)=\frac{d_{\min}(\C_{D_a})}{p-1}.\]
\end{remark}
  A monomially equivalent weighted punctured code is self-orthogonal if and only if there exists a choice of weights $(c_P)_{P\in X_a}\in \mathcal S_p^{X_a}$ such that
\begin{equation}\label{eq:pronewc}
\sum_{P\in X_a}c_Pg_Pg_P^T=0,\qquad c_P\in \mathcal S_p.
\end{equation} or equivalently,
\begin{lemma}
	Set  $A_P=g_Pg_P^T\in\Sym_m(\BF_p)$ and
	\[L_a:\BF_p^{X_a}\longrightarrow\Sym_m(\BF_p),
	\qquad
	(c_P)_{P\in X_a}\longmapsto\sum_{P\in X_a}c_PA_P.
	\]
	There exists a choice of representatives or a non-zero coordinate multiplier preserving self-orthogonality if and only if
	\begin{equation}\label{eq:pronewco2}\ker L_a\cap \mathcal S_p^{X_a}\ne\emptyset.\end{equation}
	The property does not depend on
	\begin{itemize}
		\item the basis of $\BF_{q}/\BF_p$,
		\item the initial representative chosen for each projective point.
	\end{itemize}
\end{lemma}
\begin{proof}
	\noindent \textbf{(\textperiodcentered $\Rightarrow$) Necessity.}
	Assume that there exists a choice of representatives such that the condition holds.
	Specifically, for each point $P \in X_a$, let the original representative be replaced by $\lambda_P \cdot r_P$ with $\lambda_P \in \BF_p^*$.
	Under this operation, the associated symmetric matrix transforms as $A_P \mapsto \lambda_P^2 A_P$.
	The self-orthogonality condition implies that the weighted sum of these adjusted matrices is the zero matrix
	\[
	\sum_{P \in X_a} \lambda_P^2 A_P = 0.
	\]
	
	Define the vector $(\sigma_P)_{P \in X_a}$ by setting $\sigma_P = \lambda_P^2$ for all $P$.
	Since $\lambda_P \in \BF_p^*$, we have $\sigma_P \in (\BF_p^*)^2 = \mathcal S_p$. Thus, $(\sigma_P) \in \mathcal S_p^{X_a}$.
	
	Furthermore, by definition of $L_a$, the equation $\sum\limits_{P \in X_a} \sigma_P A_P = 0$ means precisely that $(\sigma_P) \in \ker L_a$.
	Consequently, $(\sigma_P) \in \ker L_a \cap \mathcal S_p^{X_a}$, proving that the intersection is non-empty.
	
	\medskip
	\noindent \textbf{(\textperiodcentered $\Leftarrow$) Sufficiency.}
	Conversely, assume that there exists a vector $(\sigma_P)_{P \in X_a} \in \ker L_a \cap \mathcal S_p^{X_a}$.
	Since $(\sigma_P) \in \mathcal S_p^{X_a}$, for every $P \in X_a$, there exists $\lambda_P \in \BF_p^*$ such that $\sigma_P = \lambda_P^2$.
	Since $(\sigma_P) \in \ker L_a$, we have:
	\[
	L_a((\sigma_P)) = \sum_{P \in X_a} \sigma_P A_P = \sum_{P \in X_a} \lambda_P^2 A_P = 0.
	\]
	Now, choose the representative of $P$ to be $\lambda_P \cdot r_P$ (multiplying the initial representative by $\lambda_P$).
	The corresponding matrix becomes $A'_P = \lambda_P^2 A_P$,
	substituting this into the sum yields
	\[
	\sum_{P \in X_a} A'_P = \sum_{P \in X_a} \lambda_P^2 A_P = 0.
	\]
	Thus, the adjusted set of representatives satisfies the required self-orthogonality condition.
	\medskip
	
	We verify that the condition $\ker L_a \cap \mathcal S_p^{X_a} \neq \emptyset$ is independent of the choices of basis, initial representatives, and the allocation of scaling factors.
	
	\begin{enumerate}
		\item \textbf{Independence of the choice of basis of $\BF_{q}/\BF_p$.}
		Changing the basis is equivalent to left-multiplying all vectors $g_P$ by a fixed invertible matrix $U \in \mathrm{GL}_m(\BF_p)$. Consequently, the associated symmetric matrices transform as $A_P \mapsto U A_P U^{T}$.
		Under this transformation, the condition $\sum\limits_{P \in X_a} c_P A_P = 0$ becomes:
		\[
		\sum_{P \in X_a} c_P (U A_P U^{T})
		=U \left(\sum_{P \in X_a}c_PA_P\right)U^T
		=0.
		\]
		Since $U$ is invertible, the solution space of the linear equation remains exactly the same. Moreover, the condition that $c_P$ belongs to the set of non-zero squares $\mathcal S_p$ is also preserved.
		
		\item \textbf{Independence of the initial representative.}
		Replacing the initial representative $g_P$ with $\mu_P g_P$ for some $\mu_P \in \BF_p^*$ results in the transformation $A_P \mapsto \mu_P^2 A_P$.
		If $(c_P)_{P \in X_a} \in \ker L_a$, then $\sum c_P A_P = 0$. Under the new scaling, the coefficients must adjust to $c'_P = c_P / \mu_P^2$ to maintain the vanishing condition:
		\[
		\begin{aligned}
		\sum_{P \in X_a} c'_P (\mu_P^2 A_P)
		&=\sum_{P \in X_a}\left(\frac{c_P}{\mu_P^2}\right)\mu_P^2A_P\\
		&=\sum_{P \in X_a}c_PA_P=0.
		\end{aligned}
		\]
		Since $\mu_P^2 \in \mathcal S_p$, the mapping $c_P \leftrightarrow c'_P$ is a bijection between square factors. Thus, the property that $\ker L_a \cap \mathcal S_p^{X_a} \neq \emptyset$ remains invariant.
	\end{enumerate}
	Therefore, we finish the proof.
\end{proof}

\begin{lemma}
	Let
	\[
	\mathscr W_a = \operatorname{span}_{\BF_p} \{ A_P : P \in X_a \}.
	\]
	Then
	\begin{equation} \label{eq:rank-la}
		\rank(L_a) = \dim \mathscr W_a.
	\end{equation}
	
	Let $V_q$ denote the $m$-dimensional $\BF_p$-space underlying $\BF_q$. With respect to the standard non-degenerate pairing on the space of symmetric matrices $\Sym_m(\BF_p)$ (identified with the quadratic form space $\mathcal Q(V_q)$), the orthogonal complement of $\mathscr W_a$ is characterized as
	\begin{equation} \label{eq:orth-complement}
	\mathscr W_a^\perp = \{ F \in \mathcal Q(V_q):\ F(P):=F(g_P) = 0 \text{ for every } P \in X_a \}.
	\end{equation}
	
	Consequently, 
		\[
		\langle Q_0, Q_a \rangle_{\BF_p} \subseteq \mathscr W_a^\perp,
		\]
		which implies that
		\begin{equation} \label{eq:rank-bound-a-not-in-k}
		\rank(L_a) \leq \binom{m+1}{2} - 2.
		\end{equation}
\end{lemma}
\begin{proof}
	Since $L_a((c_P)_{P\in X_a})=\sum c_P A_P$, we have \[\operatorname{Im}(L_a)=\mathscr W_a,\] hence $\rank(L_a)=\dim \mathscr W_a$, proving ~\eqref{eq:rank-la}.
	
	Identify $\Sym_m(\BF_p)$ with the quadratic form space $\mathcal Q(V_q)$ via \[A\leftrightarrow F_A(x)=x^{T}Ax.\]
	
	The standard pairing   satisfies $\langle B,A_P\rangle=\operatorname{tr}(BA_P)=g_P^TBg_P=F_B(g_P)$ . Therefore,
	\[
	\begin{aligned}
	\mathscr W_a^\perp
	&=\{B:\operatorname{tr}(BA_P)=0,\ \forall P\in X_a\}\\
	&=\{F\in\mathcal Q(V_q):F(P)=0,\ \forall P\in X_a\}.
	\end{aligned}
	\]
	which is ~\eqref{eq:orth-complement}.
	
	If   $a\notin \BF_p$ , then for all $P\in X_a$, $Q_0(P)=Q_a(P)=0$. Hence, \[Q_0,Q_a\in \mathscr W_a^\perp.\] Since $Q_0$ and $Q_a$ are linearly independent, $\dim \mathscr W_a^\perp\ge2$. Using $\dim \mathscr W_a+\dim \mathscr W_a^\perp=\binom{m+1}{2}$, we get \[\rank(L_a)\le\binom{m+1}{2}-2,\] proving ~\eqref{eq:rank-bound-a-not-in-k}.
\end{proof}
If $X_a = \emptyset$, then the following holds.
\begin{itemize}
	\item The span $\mathscr W_a$ $= \operatorname{span}_{\BF_p}\{A_P \mid P\in X_a\} = 0$ .
	\item For a function space on an empty index set, there exists exactly one function (the empty map), so:
	\[
	  \BF_p ^{X_a} = \mathcal S_p^{X_a} = \{\emptyset\}
	\]
	\item The unique empty element trivially lies in $\ker L_a$, so the condition $\ker L_a \cap \mathcal S_p^{X_a} \neq \emptyset$ holds vacuously, corresponding to a length-$0$ trivial code.
\end{itemize}

If  $X_a \neq \emptyset$, then $g_P\neq0$ and $A_P=g_Pg_P^T\neq0$. Hence 
\[\mathscr W_a\neq0.\] We next assume $X_a\neq\emptyset$ and take
\[
r=\rank(L_a),\qquad h_0=\frac{p-1}{2},\qquad
\rho_p=\frac{\sqrt p+1}{p-1}<1\quad(p>3).
\]
	Define
\[
\Delta_a^{\rm supp}=
\min  _{0\ne\ell\in \mathscr W_a^*}
 \#\{P\in X_a:\ell(A_P)\ne0\}.
\]

By means of Fourier expansion, we derive strict sufficient conditions for Eq.~\eqref{eq:pronewco2} to hold.
\begin{proposition}
	We have
	\[
	(p^r-1)\rho_p^{\Delta_a^{\rm supp}}<1
	\Longrightarrow
	\ker L_a\cap \mathcal S_p^{X_a}\ne\emptyset.
	\]
\end{proposition}
\begin{proof}
	We have
	\[
	\begin{aligned}
	\sum_{c\in \mathcal S_p}\chi_p(tc)
	&=\frac{1}{2}\left(\sum_{c\neq0}\chi_p(tc)
	+\sum_{c\neq0}\eta_p(c)\chi_p(tc)\right)\\
	&=\begin{cases}
		h_0, & t=0\\
	\dfrac{-1+\eta_p(t)G_p}{2}, & t\neq0.
	\end{cases}
	\end{aligned}
	\]
	Since $|\eta_p(t)G_p|=\sqrt{p}$, then for $t\neq0$,
	\[|\sum_{c\in \mathcal S_p}\chi_p(tc)|\leq\frac{1+\sqrt{p}}{2}=h_0\rho_p.\]
	
	We next aim to count the number of $c_P\in \mathcal S_p$ satisfying $\sum\limits_{P\in X_a}c_PA_P=0$. For any linear functional $\ell\in\mathscr W_a^*$, the linearity property ensures that
	\[\ell(\sum_Pc_PA_P)=\sum_Pc_P\ell(A_P).\]
	Therefore, the indicator function for the condition $\sum_Pc_PA_P=0$ can be expressed via the orthogonality relation of characters as
	\[
	\begin{aligned}
	\mathrm{1}_{\{\sum_Pc_PA_P=0\}}
	&=p^{-r}\sum_{\ell\in \mathscr W_a^*}
	\chi_p\!\left(\ell\!\left(\sum_{P\in X_a}c_PA_P\right)\right)\\
	&=p^{-r}\sum_{\ell\in \mathscr W_a^*}
	\prod_{P\in X_a}\chi_p(c_P\ell(A_P)).
	\end{aligned}
	\]
	
	Summing over all   $(c_P)\in \mathcal S_p^{X_a}$ yields
	\begin{align*}
	N_a^{\rm sq}
	&=\sum_{(c_P)\in \mathcal S_p^{X_a}}\mathrm{1}_{\{\sum_Pc_PA_P=0\}}\\
	&=p^{-r}\sum_{\ell\in \mathscr W_a^*}\prod_{P\in X_a}
	\left(\sum_{c\in \mathcal S_p}\chi_p(c\ell(A_P))\right).
	\end{align*}
	If $\ell=0$, i.e., $\ell(A_P)=0$ for all $P\in X_a$, then
	\[\prod_{P\in X_a}\bigl(\sum  _{c\in \mathcal S_p} \chi_p(c \ell(A_P))\bigr)=h_0^s.\]
	If $\ell\neq0$, by the definition of $\Delta_a^{\rm supp}$, for any non-zero $\ell$, there are at least $\Delta_a^{\rm supp}$ points $P$ such that $\ell(A_P)\neq0$. For the remaining $s-\Delta_a^{\rm supp}$ points, we have $\ell(A_P)=0$, which contributes a factor of $h_0$.
	
	For any non-zero $t=\ell(A_P)$, we have the bound $|\sum_{c\in \mathcal S_p}\chi_p(tc)|\leq h_0\rho_p$  .  Therefore, the modulus of the product term for a non-zero $\ell$ satisfies
	\[
	\begin{aligned}
	\left|\prod_{P\in X_a}\left(\sum_{c\in \mathcal S_p}\chi_p(c\ell(A_P))\right)\right|
	\le h_0^{s-\Delta_a^{\rm supp}}(h_0\rho_p)^{\Delta_a^{\rm supp}}
	=h_0^s\rho_p^{\Delta_a^{\rm supp}}.
	\end{aligned}
	\]
	Applying the triangle inequality to the expression for $N_a^{\rm sq}$ yields
	\begin{align*}
	N_a^{\rm sq}&\ge p^{-r}\left(h_0^s-\sum_{\ell\neq0}
	\left|\prod_{P\in X_a}\left(\sum_{c\in \mathcal S_p}\chi_p(c\ell(A_P))\right)\right|\right)\\
	&\ge p^{-r}\bigl(h_0^s-(p^r-1)h_0^s\rho_p^{\Delta_a^{\rm supp}}\bigr)\\
	&=p^{-r}h_0^s\bigl(1-(p^r-1)\rho_p^{\Delta_a^{\rm supp}}\bigr).
	\end{align*}This completes the proof.
\end{proof}

For $m\geq4$, we have
\begin{proposition}
	Let $X_a\subseteq\mathbb P^{m-1}(\BF_p)$ be the projective point set defined by the trace quadratic forms $Q_0$ and $Q_a$. Define
	\[
	 \pi_j= \frac{p^{j+1}-1}{p-1},\qquad r_{\max}=\binom{m+1}{2}-2.
	\]
	 If $a\notin \BF_p$, $d=m$, and $m\ge 4$, then
		\[
		(p^{r_{\max}}-1)\,\rho_p^{\,s_*-8\pi_{m-4}}<1
		\quad\Longrightarrow\quad
		\ker L_a\cap \mathcal S_p^{X_a}\neq\emptyset,
		\]
		where $r_{\max}=\binom{m+1}{2}-2$, $s_*=\pi_{m-3}$ for odd $m$, and
		$s_*=\pi_{m-3}-(p+1)p^{m/2-2}$ for even $m$.
\end{proposition}
\begin{proof}
	
	Over $\overline{\BF}_p$, $Q_0$ and $Q_a$ can be simultaneously diagonalized with pairwise distinct coefficients, so $\mathscr X_a$ is a smooth complete intersection of dimension $m-3$.
	For any quadric $\mathcal R\notin\langle Q_0,Q_a\rangle$, the intersection $V(Q_0,Q_a,\mathcal R)$ has dimension $m-4$ and degree $2\cdot2\cdot2=8$.
	Thus
	\[
	\#\{P\in X_a : \mathcal R(P)=0\}\le 8\pi_{m-4},
	\]
	and consequently $\Delta_a^{\rm supp}\ge s-8\pi_{m-4}$.
	For odd $m$, $s=\pi_{m-3}$, for even $m$, Lemma~\ref{n0} gives $s\ge s_*=\pi_{m-3}-(p+1)p^{m/2-2}$.
	Together with the rank bound \[r\le r_{\max}=\binom{m+1}{2}-2\] from Eq.~\eqref{eq:rank-bound-a-not-in-k}, the condition
	\[
	(p^{r_{\max}}-1)\rho_p^{\,s_*-8\pi_{m-4}}<1
	\]
	implies $(p^r-1)\rho_p^{\Delta_a^{\rm supp}}<1$, hence $\ker L_a\cap \mathcal S_p^{X_a}\neq\emptyset$.
	
	\smallskip

	In above case the exponent grows like  $p^{m-4}(p-8)$, while $r_{\max}=O(m^2)$.
	Thus for $p\ge 11$, the inequalities hold for all sufficiently large $m$.
\end{proof}
For the remaining low-dimensional cases, we have
\begin{proposition} \label{pro:low-dim-classification}
	 The following classification holds completely.
	
	\begin{enumerate}[label=\textbf{Case 1:}, leftmargin=*, nosep]
		\item \textbf{$p=3$.}
		\begin{itemize}[nosep]
			\item If $a \notin \mathbb{F}_3$, $D_a\ne\emptyset$, and $m$ is odd, then Eq.~\eqref{eq:pronewc} is solvable precisely when $m \ge 7$.
			\item If $a \notin \mathbb{F}_3$, $D_a\ne\emptyset$, and $m$ is even, then Eq.~\eqref{eq:pronewc} is solvable precisely when $m \ge 6$.
		\end{itemize}
	\end{enumerate}
	
	\begin{enumerate}[label=\textbf{Case 2:}, leftmargin=*, nosep, resume]
		\item \textbf{$p > 3$.}
		\begin{itemize}[nosep]
			\item If $m=2$ and $a \notin \BF_p$, then $D_a = \emptyset$, and only the empty code exists.
			\item If $m=3$ and $a \notin \BF_p$, then the punctured code $\C_{R_a}$ is a $[1,1,1]_p$ code, which is never solvable.
		\end{itemize}
	\end{enumerate}
\end{proposition}
\begin{proof} \textbf{Case 1.} When $p=3$ and $a\notin\BF_3$, by Remark~\ref{re:r3}, the code $\C_{R_a}$ is self-orthogonal exactly when $2\nmid m$ and $m\geq7$ or $2\mid m$ and $m\geq6$ by Theorem~\ref{th:sef}.

\smallskip

\textbf{Case 2.} For $m=2$, by Proposition~\ref{pro:dim}, $D_a=\emptyset$. For $m=3$, the code $\C_{D_a}$ has parameters $[p-1,1,p-1]_p$, then the code $\C_{R_a}$ has parameters $[1,1,1]_p$.
\end{proof}

We now complete the remaining classification of when the orbit-punctured code $\C_{R_a}$ can be made Euclidean self-orthogonal via a suitable choice of representatives or non-zero coordinate multipliers. We first precisely determine the rank of the linear map $L_a$ and characterize its dual kernel.
\begin{definition}
	Let $\mathcal Q(V_q)$ be the space of quadratic forms on $V_q=\BF_q$, which is identified with $\Sym_m(\BF_p)$ via $F_A(x)=x^TAx$. For the fixed representative vector $g_P$, put $F(P):=F(g_P)$. Let \[
	\operatorname{ev}_{X_a}:\mathcal Q(V_q)\longrightarrow \BF_p^{X_a},
	\qquad F\longmapsto(F(P))_{P\in X_a}.
	\]The vanishing ideal of degree $2$ is
	\begin{equation*}
	I_2(X_a):=\{F\in\mathcal Q(V_q):F(P)=0,\ \forall P\in X_a\},
	\end{equation*}
\end{definition}
\begin{lemma}
	With respect to the standard non-degenerate pairing $
	\langle B,A\rangle=\operatorname{tr}(BA)$ between $\Sym_m(\BF_p)$ and $\mathcal Q(V_q)$, the adjoint of $L_a$ is exactly
	\[
	L_a^*=\operatorname{ev}_{X_a}.
	\]  Consequently,
\begin{equation*}
	\ker L_a^*
	=I_2(X_a)
	\end{equation*} and
	\begin{equation}\label{eq:comsq}
	\rank(L_a)
	=\rank(\operatorname{ev}_{X_a})
	=\binom{m+1}{2}-\dim I_2(X_a).
\end{equation}
	In particular,
	\begin{itemize}
		\item if $a\notin \BF_p$, then $Q_0$ and $Q_a$ are linearly independent and both belong to $I_2(X_a)$,
		\item if $X_a=\emptyset$, then $I_2(X_a)=\mathcal Q(V_q)$ and  $\rank(L_a)=0$.
	\end{itemize}
\end{lemma}
\begin{proof}
	For any symmetric matrix $B$ corresponding to the quadratic form $F_B(x)=x^TBx$, the trace pairing satisfies
	\[
	\langle B,g_Pg_P^T\rangle
	=\operatorname{tr}(Bg_Pg_P^T)
	=g_P^TBg_P=F_B(P).
	\]
	Therefore, for any coefficient vector $(c_P)\in\BF_p^{X_a}$,
	\[
	\begin{aligned}
	\langle B,L_a((c_P))\rangle
	&=\sum_{P\in X_a}c_P\langle B,g_Pg_P^T\rangle\\
	&=\sum_{P\in X_a}c_PF_B(P)\\
	&=\langle\operatorname{ev}_{X_a}(F_B),(c_P)\rangle.
	\end{aligned}
	\]
	By the definition of the adjoint map, this gives $L^*_a(F_B)=\operatorname{ev}_{X_a}(F_B)$, hence
	\[L_a^*=\operatorname{ev}_{X_a}.\]
	
	Since $\rank(L_a)=\rank(L_a^*)=\rank(\operatorname{ev}_{X_a})$, Eq.~\eqref{eq:comsq} follows from the rank-nullity theorem applied to $\operatorname{ev}_{X_a}$
	\[
	\begin{aligned}
	\dim\mathcal Q(V_q)
	&=\rank(\operatorname{ev}_{X_a})+\dim\ker\operatorname{ev}_{X_a}\\
	&\Longrightarrow\quad
	\rank(\operatorname{ev}_{X_a})
	=\binom{m+1}{2}-\dim I_2(X_a).
	\end{aligned}
	\]
	For the two special cases:
	\begin{itemize}
		\item $a\notin \BF_p$, both $Q_0(P)=0$ and $Q_a(P)=0$ hold for all $P\in X_a$. Since $a\notin \BF_p$, $Q_0$ and $Q_a$ are linearly independent, so both lie in $I_2(X_a)$.
		\item $X_a=\emptyset$, the condition $F(P)=0$ is vacuously true for every $F\in\mathcal Q(V_q)$, so $I_2(X_a)=\mathcal Q(V_q)$ and $\rank(\operatorname{ev}_{X_a})=0$, giving $\rank(L_a)=0$.
	\end{itemize}
	Finally, since the pairing and the adjoint depend only on the vector space structure, the conclusions are independent of the choice of basis of $\BF_q/\BF_p$ and of the initial representatives $g_P$.
\end{proof}

With $h_0$ as above, let $t \ge m$ be an integer satisfying
\begin{equation}\label{eq:t-cond}
t \equiv -1 \pmod{h_0}.
\end{equation}
Let $E = \mathbb{F}_{p^t}$ be a finite field extension of $\BF_p$ of degree $t$, and let $\iota : \BF_{p^m} \to E$ be an arbitrary injective $\BF_p$-linear map. Define the homogeneous polynomial
\begin{equation}\label{eq:norm-poly}
\nu_\iota(x) = N_{E/\BF_p}(\iota(x)),
\end{equation}
where $N_{E/\BF_p}$ denotes the field norm from $E$ to $\BF_p$.
\begin{theorem}\label{th:norcon}
	  Assume  the conditions of ~\eqref{eq:t-cond} and ~\eqref{eq:norm-poly}  hold. If $a\notin\BF_p$ and
		 \begin{equation}\label{eq:cond-a-in-k} \, 2t+2+4(p-1)<m(p-1),\end{equation}
		then for every $P\in X_a$ (with representatives $r_P\in\BF_q$), the weights
		 \[\gamma_P=\nu_\iota(r_P),\qquad c_P=\nu_\iota(r_P)^2\]
		  satisfy  $\gamma_P\in\BF_p^*$ and $c_P\in \mathcal S_p$ , and hence $\ker L_a\cap \mathcal S_p^{X_a}\neq\emptyset$.
	 \end{theorem}
\begin{proof}
	We recall an important result, if $f \in \BF_p[X_1, \ldots, X_m]$ with $\deg f < m(p-1)$, then \[\sum_{x \in \BF_p^m} f(x) = 0.\]
	
	Since $\iota$ is injective, \[x \neq 0 \implies \iota(x) \neq 0 \implies \nu_\iota(x) \neq 0.\] Thus, $c_P = \nu_\iota(r_P)^2$ is always a non-zero square. For any $\lambda\in\BF_p$,
	\[
	\begin{aligned}
	N_{E/\BF_p}(\iota(\lambda x))
	&=N_{E/\BF_p}(\lambda\iota(x))\\
	&=\lambda^tN_{E/\BF_p}(\iota(x)).
	\end{aligned}
	\]We further have
	\[\nu_\iota(\lambda x)=\lambda^t\nu_\iota(x).\]
	Hence, $\nu_\iota(x)$ is a $t$-homogeneous function over $\BF_p$.
	
	\medskip
	We need to show $\sum\limits_{P \in X_a} c_P g_P g_P^{T} = 0$.
		When $a \notin \BF_p$, define \[I_a(x) = (1 - Q_0(x)^{p-1})(1 - Q_a(x)^{p-1}).\] This is the indicator function of $D_a\cup\{0\}$ and $\deg(I_a(x))\le 4(p-1)$. Since $2t+2+4(p-1)<m(p-1)$, for any coordinate linear forms $\ell_i, \ell_j$, set \[F_{ij}(x)=\nu_\iota(x)^2\ell_i(x)\ell_j(x)I_a(x).\]
	Then $\deg F_{ij}(x)\le 2t+2+4(p-1)<m(p-1)$, and thus the sum \[\sum_{x\in \BF_q}F_{ij}(x) = 0.\]
		
		Note that every $x\in D_a$ can be written as $x=\lambda r_P$, where $P=[r_P]\in X_a$. Since $I_a$ vanishes on $\BF_q\setminus(D_a\cup\{0\})$, we have
		\[0=\sum_{x\in \BF_q}F_{ij}(x) = F_{ij}(0)+\sum_{P\in X_a}\sum_{\lambda\in\BF_p^*}F_{ij}(\lambda r_P).\]
		Since $F_{ij}(0)=\nu_\iota(0)^2=0$, $\nu_\iota(\lambda r_P)=\lambda^t\nu_\iota(r_P)$ and $\ell_i(\lambda x)=\lambda\ell_i(x)$, we get
		\begin{align*}
		F_{ij}(\lambda r_P)
		&=(\lambda^t\nu_\iota(r_P))^2(\lambda\ell_i(r_P))
		(\lambda\ell_j(r_P))\\
		&=\nu_\iota(r_P)^2\ell_i(r_P)\ell_j(r_P)\lambda^{2t+2}\\
		&=c_P\ell_i(r_P)\ell_j(r_P)\lambda^{2t+2}.
		\end{align*}
		Therefore, we get
		\[
		\left(\sum_{P\in X_a}c_P\ell_i(r_P)\ell_j(r_P)\right)
		\left(\sum_{\lambda\in\BF_p^*}\lambda^{2t+2}\right)=0.
		\]
		By ~\eqref{eq:t-cond}, $p-1 \mid 2t+2$, so $\sum\limits_{\lambda \in \BF_p^*} \lambda^{2t+2} = p-1 = -1 \neq 0$ in $\BF_p$. This forces the matrix entry $\sum\limits_{P \in X_a} c_P \ell_i(r_P) \ell_j(r_P) = 0$, which is exactly ~\eqref{eq:pronewc} in matrix form.
	Therefore,   \[(c_P)\in \ker L_a\cap \mathcal S_p^{X_a}.\]  The proof is completed.
\end{proof}
We next give   a  unified infinite existence range via the norm construction.
\begin{theorem}
	Adopt the notations of Theorem~\ref{th:norcon} and let $t$ be the smallest integer satisfying
	\[t\ge m,\qquad t\equiv-1\pmod{(p-1)/2}.\]
	Then there exists a non-trivial square-weight vector   $(c_P)\in \ker L_a\cap \mathcal S_p^{X_a}$ , and consequently the weighted punctured code $\C$ is self-orthogonal for the following ranges of  $a\notin\BF_p$.
	\[
  m\geq11\ \text{ if }\ p=5,\  m\ge 8\ \text{ if }\ p=7,\ \text{ and }\ m\ge7\ \text{ if }\ p\ge11.\]
		
	This construction is independent of $d=[\BF_p(a):\BF_p]$ and remains valid even when geometric degeneration occurs.
\end{theorem}
\begin{proof}
	Let $t$ be the smallest integer satisfying the congruence. Then
	\[m\le t\le m+(p-3)/2\implies 2t+2\le 2m+p-1.\]

		 For $a\notin\BF_p$, the left-hand side of condition ~\eqref{eq:cond-a-in-k} is at most $2m+5(p-1)$. It suffices to have
		\[m(p-3)>5(p-1),\]which gives the stated lower bounds on  $m$.

	The estimates involve no dependence on $d$. When $d=m$, $\mathscr X_a$ is a smooth complete intersection of two quadrics, when $d<m$, the geometry may degenerate, but the construction uses only the two equations over the finite field and remains effective.
\end{proof}
\begin{example}
Take $p=5$, $m=7$, and $a=1\in\BF_5$. Realize
$E=\BF_5[X]/(X^7+X^6+1)$ and take $\iota$ to be the identity map. With
respect to the polynomial basis $(1,X,\ldots,X^6)$, choose the unique
representative of each projective point whose first nonzero coordinate is
$1$. Then $t=7$ is the smallest integer satisfying
\[
t\ge m,\qquad t\equiv -1\pmod{(p-1)/2}.
\]
An exhaustive computation gives $\#X_a=3906$. For
$c_P=\nu_\iota(r_P)^2=N_{E/\BF_5}(r_P)^2$, the program returns
\[
\#\{P:c_P=1\}=1988,\qquad \#\{P:c_P=4\}=1918,
\]
and
\[
L_a((c_P))=\sum_{P\in X_a}c_Pg_Pg_P^T=0_{7\times7}.
\]
Thus $(c_P)\in\ker L_a\cap \mathcal S_p^{X_a}$ and the resulting weighted punctured
code is self-orthogonal.
\end{example}

\subsection{Dual distance}
\begin{theorem}
	Let $\C_{D_a}$ be a code over $\BF_p$ with defining set $D_a$ and $D_a\neq\emptyset$.
	Then
	\[
	d_{\min}(\C_{D_a}^{\perp})=2,\qquad A_2(\C_{D_a}^{\perp})
	=\frac{n(p-1)(p-2)}2.
	\]
\end{theorem}
\begin{proof}
	Choose a $\BF_p$-basis $e_1,e_2,\ldots,e_m$ of $\BF_q$. For any $x\in D_a$, its corresponding column of the generator matrix is
	\[
	g(x)=\bigl(\Tr(e_1x),\dots,
	\Tr(e_mx)\bigr)^T.
	\]For any $t\in\BF_p^*$, we get $g(tx)=tg(x)$  ,  this means any two distinct coordinates yield non-zero proportional columns   at the same projective point $[x]\in\mathbb P^{m-1}(\BF_p)$, and they therefore determine a weight-2 dual codeword .
	
	In addition, columns belonging to distinct $\BF_p^*$-orbit are not proportional. Indeed, if  $g(x)=\lambda g(y)$ for some $\lambda\in\BF_p^*$, then
	\[
	\Tr(b(x-\lambda y))=0\quad\text{for all }b\in \BF_q,
	\]Non-degeneracy of the trace pairing forces $x-\lambda y=0$, i.e., $x=\lambda y$, so $x$ and $y$ lie in the same orbit. Consequently, every 2-weight dual codeword arises from exactly two coordinates within the same orbit.
	
	Each orbit contains $p-1$ coordinates, hence contributes $\binom{p-1}{2}$ supports. Then
	\[
	A_2(\C_{D_a}^{\perp})
	=\frac n{p-1}\binom{p-1}{2}(p-1)
	=\frac{n(p-1)(p-2)}2.
	\]
\end{proof}
\begin{example}
Take $p=5$, $m=6$, and let $a\in\BF_{5^3}\subset\BF_{5^6}$ be a root of
$T^3+2T+1\in\BF_5[T]$. An exhaustive computation gives $n=504$ and no
zero column in a generator matrix. Its columns form $126$ projective classes,
each of size $4$, and hence determine $756$ dependent two-column supports.
The program therefore returns
\[
d_{\min}(\C_{D_a}^{\perp})=2,\qquad
A_2(\C_{D_a}^{\perp})=756\cdot4=3024.
\]
\end{example}
Set \[
X_a=\{[x]\in\mathbb P^{m-1}(\BF_p):
\Tr(x^2)=\Tr(ax^2)=0\}.
\]
Since each $\BF_p^*$-orbit corresponds to a unique projective point, so the columns of a generator matrix of $\C_{R_a}$ are precisely the $\BF_p$-rational points of $X_a$ and no two columns are linearly dependent. Therefore,
\[d_{\min}(\C_{R_a}^{\perp})\ge3.\]
\begin{theorem}\label{th:dual3}
	Denote by $\Lambda_a$ the number of $\BF_p$-rational   lines  in $X_a$. Then the number of 3-weight   codewords  in $\C_{R_a}^\perp$ is
	\[A_3(\C_{R_a}^{\perp})
	=(p-1)\binom{p+1}{3}\Lambda_a.\]
\end{theorem}
\begin{proof}
	Three non-zero, pairwise non-proportional columns are linearly dependent precisely if and only if they lie on a projective line, which is equivalent to there exist a 2-dimensional vector subspace  $W=\langle x,y\rangle_{\BF_p}\subset \BF_q$ such that \[
	\Tr(w^2)=\Tr(aw^2)=0
	\qquad(w\in W).
	\]Equivalently,
	\begin{equation}\label{eq:6}
	\begin{aligned}
	&\Tr(x^2)=\Tr(y^2)
	=\Tr(xy)=0,\\
	&\Tr(ax^2)=\Tr(ay^2)
	=\Tr(axy)=0.
	\end{aligned}
	\end{equation}
	
	Note that a projective line over $\BF_p$ contains exactly $p+1$ points, hence $\binom{p+1}{3}$ triples of distinct points. The corresponding three columns span a two-dimensional space, so their linear-relation space is one-dimensional; accounting for the
	$p-1$ non-zero scalar multiples gives the factor $(p-1)\binom{p+1}{3}$ 3-weight dual codeword over any $\BF_p$-rational line. Summing over all such lines yield $A_3(\C_{R_a}^{\perp})$.
\end{proof}
By Theorem~\ref{th:dual3}, we have
\[d_{\min}(\C_{R_a}^{\perp})=3
\iff \Lambda_a>0.\]

We next give the complete classification of punctured dual distances. Take $U=\langle1,a\rangle_{\BF_p}\subset \BF_q$. For any $(A,B,C)\in U^3$, we define
the binary-polynomial discriminant
\[
\Delta(A,B,C)=B^2-4AC.
\]
Set
\[\mathcal E_2(a)=
\sum_{\substack{(A,B,C)\in U^3\\\Delta\ne0}}
\eta_q(\Delta),\qquad
\mathcal E_1(a)=
\sum_{\substack{(A,B,C)\in U^3\setminus\{0\}\\\Delta=0}}
\tau(A,B,C),\]where\[
\tau(A,B,C)=
\begin{cases}
\eta_q(A),&A\ne0,\\
\eta_q(C),&A=0.
\end{cases}
\]

Clearly,
\[
\begin{aligned}
\bigl|\mathcal E_2(a)\bigr|&\le p^6-1-\#\{(A,B,C)\in U^3\setminus\{0\}:\Delta=0\},\\
\bigl|\mathcal E_1(a)\bigr|&\le \#\{(A,B,C)\in U^3\setminus\{0\}:\Delta=0\}.
\end{aligned}
\]
\begin{lemma}
	Let $n_0=n+1$ be the number of affine common zeros of the two quadratic equations. Then the number of lines contained in $X_a$ is
	\[
	\Lambda_a=
	\frac{p^{-6}\bigl(q^2+q\mathcal E_2(a)+qG_q\mathcal E_1(a)\bigr)-1-(p+1)(n_0-1)}{(p^2-1)(p^2-p)}
	\]
	
	Moreover, the number of non-zero triples satisfying $\Delta=0$ is exactly\[
	R_{\rm rk1}=
	\begin{cases}
	p^4-1,&d=2,\\
	(p+1)(p^2-1),&d\ne2.
	\end{cases}
	\]
\end{lemma}
\begin{proof}
	Denote by $\mathcal M_a$ the number of common solutions to the six trace conditions in ~\eqref{eq:6}. For each $(A,B,C)\in U^3$, define the inner character sum
	\[
	T(A,B,C):=\sum_{x,y\in\BF_q}\chi_q(Ax^2+Bxy+Cy^2).
	\]
	By additive character orthogonality,
\begin{align*}
	\mathcal M_a
	&=p^{-6}\sum_{l_1,\ldots,l_6\in\BF_p}\sum_{x,y\in\BF_q}
	\chi_q\bigl(l_1x^2+l_2y^2+l_3xy+l_4ax^2+l_5ay^2+l_6axy\bigr)\notag\\
	&=p^{-6}\sum_{(A,B,C)\in U^3}T(A,B,C),
	\end{align*}
	where $A=l_1+l_4a$, $B=l_3+l_6a$, and $C=l_2+l_5a$.
	
	If $\Delta\neq0$, the quadratic form is of full rank over $\BF_q^2$. When $A\neq0$, take $u=x+\frac{B}{2A}y$. Then
	\begin{align}\label{eq:m6}
	T(A,B,C)
	&=\left(\sum_{u\in\BF_q}\chi_q(Au^2)\right)
	\left(\sum_{y\in\BF_q}
	\chi_q\left(\left(C-\frac{B^2}{4A}\right)y^2\right)\right)\notag\\
	&=\eta_q(A)G_q\eta_q\left(\frac{4AC-B^2}{4A}\right)G_q\notag\\
	&=\eta_q(-1)\eta_q(\Delta)G_q^2
	=q\eta_q(\Delta).
	\end{align}
	If $A=0$ and $\Delta\neq0$, then $B\neq0$ and the same value follows directly: $T(0,B,C)=q=q\eta_q(B^2)$.
	
	Suppose $\Delta=0$ and $(A,B,C)\neq(0,0,0)$. The quadratic form has rank $1$. If $A\neq0$, Eq.~\eqref{eq:m6} gives
	\begin{equation}\label{eq:m66}
	T(A,B,C)=q\eta_q(A)G_q.
	\end{equation}
	If $A=0$, then $B=0$ and $C\neq0$, so similarly
	\[
	T(0,0,C)=q\eta_q(C)G_q.
	\]
	Finally, $T(0,0,0)=q^2$. Therefore,
	\begin{equation}\label{eq:m7}\mathcal M_a=p^{-6}\bigl(q^2+q\mathcal E_2(a)+qG_q\mathcal E_1(a)\bigr).\end{equation}
	
	The number of linearly dependent solution pairs $(x,y)$ is
	\[
	1+(p+1)(n_0-1),
	\]
	where the zero pair is counted once, and each non-zero common zero vector determines a class of dependent pairs via $[s:t]\in\mathbb P^1(\BF_p)$. Every 2-dimensional totally singular subspace has $\#\GL_2(\BF_p)$ ordered bases. Subtracting the dependent pairs from $\mathcal M_a$ and dividing by $\#\GL_2(\BF_p)$ gives $\Lambda_a$.
	
	\medskip
	
	Suppose $A\neq0$ and $r=B/2A$. The condition $\Delta=0$ gives \[A,Ar,Ar^2\in U\implies A(1,r,r^2)\in U.\]
	
	If $r\notin \BF_p$, then $1$ and $r$ are linearly independent. Since
	\[
	A\langle1,r,r^2\rangle_{\BF_p}\subseteq U
	\quad\text{and}\quad \dim_{\BF_p}U=2,
	\]
	we obtain
	\[
	\dim_{\BF_p}\langle1,r,r^2\rangle_{\BF_p}=2.
	\]
	Thus $r$ has degree two over $\BF_p$, so $\BF_p(r)=\BF_{p^2}$ and $U=A\BF_{p^2}$. Since $1\in U$, we have $1=Az$ for some $z\in\BF_{p^2}$, and hence $A\in\BF_{p^2}^*$. Therefore $U=\BF_{p^2}$. As $U=\langle1,a\rangle_{\BF_p}$ and $a\notin\BF_p$, it follows that $d=2$. The number of non-zero rank-one symmetric $2\times2$ matrices is then
	\[
	\#\mathbb P^1(\BF_{p^2})\times(p^2-1)=(p^2+1)(p^2-1)=p^4-1.
	\]

	If $d\neq2$, then $r\in\BF_p$ and $A\in U, C=Ar^2\in U$.
	\begin{itemize}
		\item If $A\neq0$ and $r\in\BF_p$, $A$ has $p^2-1$ choices, $r$ has $p$ choices and $C=Ar^2$ is determined $\implies p(p^2-1)$.
		\item If $A=0$, then $\Delta=0\implies B=0$; the triples $(0,0,C)$ with $C\in U^*$ form the infinite-slope case. Hence $C$ has $p^2-1$ choices.
	\end{itemize}
	Therefore, the total number is $(p+1)(p^2-1)$.
The proof is finished.\qedhere
\end{proof}

	
	
\begin{theorem}[Complete classification of punctured dual distances]
	The following conclusions hold regarding the code $\C_{R_a}^\perp$.
	\begin{enumerate}[1]
		\item For all $m\ge7$, odd prime $p$ and $a\notin\BF_p$, we have
		\[
		\Lambda_a > 0,\quad	d_{\min}(\C_{R_a}^\perp) = 3.
		\]
		
		\item For $m = 5$, it must be that $d = 5$, and
		\[
		\Lambda_a = 1, \quad d_{\min}(\C_{R_a}^\perp) = 3.
		\]
		
		\item For $m = 6$ with $d \in \{2, 3, 6\}$, all three layers satisfy
		\[
		\Lambda_a > 0, \quad d_{\min}(\C_{R_a}^\perp) = 3.
		\]
		Here, for $d=3$, there is an explicit formula independent of $a$
		\[
		\Lambda_a =
		\begin{cases}
		p^2 - p + 1, & p \equiv 1 \pmod{4}, \\
		3(p^2 + p + 1), & p \equiv 3 \pmod{4}.
		\end{cases}
		\]
		Consequently, the layer \[A_3(\C_{R_a}^\perp) = (p-1)\binom{p+1}{3}\Lambda_a\] is also explicit.
		
		\item For $m = 4$.
		\begin{itemize}
			\item If $d = 2$, then $X_a=\mathscr X_a(\BF_p) = \emptyset$, and the punctured code is the zero code.
			\item If $d = 4$, let $s = \#X_a=\#\mathscr X_a(\BF_p) = p + 1 - h(a)$. Then
			\[
			d_{\min}(\C_{R_a}^\perp) =
			\begin{cases}
			+\infty, & s = \dim(\C_{R_a}), \\
			4, & s > \dim(\C_{R_a}).
			\end{cases}
			\]
		\end{itemize}
		Combining with the dimension classification in Proposition~\ref{pro:dim}, the specific cases for low values of $d$ and $\Lambda_a$ are listed in Table \ref{tab:low_dim_cases}.
		
		\begin{table}[h]
			\centering
			\caption{Specific cases for low dimensions ($p \ge 3$)}
			\label{tab:low_dim_cases}
			\begin{tabular}{c|cc}
				\hline
				$p$ & $s$ & $d_{\min}(\C_{R_a}^\perp)$ \\
				\hline
				3 & 2, 4 & $+\infty$ \\
				3 & 6 & 4 \\
				5 & 2, 4 & $+\infty$ \\
				5 & 6, 8, 10 & 4 \\
				7 & 4 & $+\infty$ \\
				7 & 6, 8, 10, 12 & 4 \\
				\hline
			\end{tabular}
		\end{table}
		For $p \ge 11$, we have $\dim(\C_{R_a}) = 4$ and $s > 4$, so the dual distance is constantly 4.
		
		\item For $m = 3$, $X_a=\mathscr X_a(\BF_p)$ consists of only one point, the punctured code is $[1, 1, 1]_p$, and its dual is the zero code. For $m = 2$, $X_a=\mathscr X_a(\BF_p) = \emptyset$.
	\end{enumerate}
\end{theorem}
\begin{proof}We proceed by induction on $m$ and case analysis based on the dimension of the quadratic form.
	
	\paragraph{Case  $m$  odd}
	By Lemma~\ref{quadratic}, we know $n_0=p^{m-2}$ and the quadratic form is of full rank over $\BF_q^2$. By Eq.~\eqref{eq:m7},
	\[
	|\mathcal E_2(a)|\le p^6-1-R_{\rm rk1},\]this implies that
	\begin{equation}\label{eq:22}\mathcal M_a\ge p^{2m-6}-(p^6-1-R_{\rm rk1})p^{m-6}.
	\end{equation}
	
	When $2\nmid m$ and $m\ge 7$, replacing $R_{\rm rk1}=(p+1)(p^2-1)$ into ~\eqref{eq:22}, then
	\[
	\mathcal M_a>1+(p+1)(p^{m-2}-1).
	\]
	Therefore, there exists an independent solution, which implies $\Lambda_a > 0$.
	
	\paragraph{Case $m$ even and $m \ge 10$}
	Using  the following bound in Eq.~\eqref{eq:m7} \[
	|\mathcal E_2(a)|\le p^6-1-R_{\rm rk1},\qquad |\mathcal E_1(a)|\le R_{\rm rk1},
	\qquad |G_q|=p^{m/2},
	\]
	and by Lemma~\ref{n0}, \[n_0=p^{m-2}-\delta(p-1)p^{m/2-2}h(a) \le p^{m-2} + (p^2-1)p^{m/2-2},\] we similarly obtain an independent solution.
	
	For $m=8, d \ne 2$, substituting $R_{\rm rk1} = (p+1)(p^2-1)$ keeps the inequality strict.
	
	For $m=8, d=2$, regard $\BF_q$ as a four-dimensional space over $K = \BF_{p^2}$. The non-degenerate $K$-quadratic form
	\[ z \longmapsto \Tr_{\BF_q/K}(z^2) \]
	has dimension four over the finite field, and possesses a non-zero isotropic vector $x$. Hence, $Kx$ as a two-dimensional $\BF_p$-subspace satisfies $\Tr_{\BF_q/K}(w^2) = 0$, thereby simultaneously satisfying the two $\BF_p$-trace constraints. This completes the proof for all $m \ge 7$.
	\smallskip
	
	\paragraph{Case $m = 6$}
	If $d=[\BF_p(a):\BF_p]=2$, then $[\BF_{p^6}:K]=3$, meaning $\BF_{p^6}$ is a 3-dimensional vector space over $K$.
	
	A desired 2-dimensional $\BF_p$-subspace of $\BF_{p^6}$ is exactly a 1-dimensional $\BF_{p^2}$-subspace $W=K\cdot x$, since $K$ itself is 2-dimensional over $\BF_p$.
	
	For any $tx\in W (t\in K)$, the trace constraint becomes
	\[\Tr((tx)^2)=\Tr_{K/\BF_p}(t^2\Tr_{\BF_q/K}(x^2)).\]
	This vanishes identically for all $t\in K$ iff $\Tr_{\BF_q/K}(x^2)=0$. Note that $\Tr_{\BF_q/K}(x^2)$ is a non-degenerate quadratic form over $K$, and a non-degenerate quadratic form on an odd-dimensional vector space of dimension $\geq3$  always has non-zero isotropic vectors. Thus there exists nonzero $x\in\BF_{p^6}$ satisfying the constraint, so $W=Kx$  is the required 2-dimensional $\BF_p$-subspace, giving $\Lambda_a\geq1$.
	
	\medskip
	
	If $d=3$, then $K = \mathbb{F}_{p^3}, e=2$. For every $z \in K^*$, \[\eta_{q}(z) = \eta_K(z)^2 = 1.\] Therefore,
	\[
	\mathcal E_2(a)=p^6-1-R_{\rm rk1},\qquad \mathcal E_1(a)=R_{\rm rk1},
	\qquad R_{\rm rk1}=(p+1)(p^2-1).
	\]
	Since $\theta=\eta_p(-1)$, $G_q = -\theta p^3$ and $n_0 = p^4 - \theta p(p^2-1)$. Substituting them into $\Lambda_a$ yields
	{\small
	\begin{equation}
	\Lambda_a=\frac{2p^6-2
	-(1+\theta p^3)(p+1)(p^2-1)
	-(p+1)\bigl(p^4-\theta p(p^2-1)-1\bigr)}{(p^2-1)(p^2-p)}.
	\end{equation}
	}
	When $\theta=1$, the numerator factors as \[(p^2-1)(p^2-p)(p^2-p+1).\] When $\theta=-1$, the numerator is \[3(p^2-1)(p^2-p)(p^2+p+1).\]
This proves $\Lambda_a$.	

\medskip

	If $d=6$, the two quadratic forms over $\overline{\BF}_{p^6}$ can be written as \[Q_0 = \sum_{i=1}^6 X_i^2\ \text{ and }\ Q_a = \sum_{i=1}^6 a_i X_i^2,\] where $a_1, \dots, a_6$ are distinct conjugates.
	
	The Jacobian matrix of the system $\{Q_0=0,Q_a=0\}$ has row equal to the gradients
	\[
	\nabla Q_0=(2X_1,2X_2,\ldots,2X_6),\]and
	\[\nabla Q_a=(2a_1X_1,2a_2X_2,\ldots,2a_6X_6).
	\]
	
	Suppose $P=[X_1:\cdots : X_6]\in \mathscr X_a$ is a singular point, so the two gradients are linearly dependent. Then there exists $\lambda\in \overline{\BF}_p$ such that for all $1\leq i\leq6$,
	\[2X_i=\lambda 2a_iX_i\implies X_i(1-\lambda a_i)=0.\]
	Since all $a_i$ are pairwise distinct, the equation $1=\lambda a_i$ holds for at most one index $i$, so at most one coordinate $X_i$ is nonzero, all others are $0$. But then $Q_0(P)=X_i^2$ forces $X_i=0$, so $P=0$, impossible in projective space.
	
	Hence, $\mathscr X_a\subset \mathbb P^5_{\BF_p}$ is a smooth complete intersection of two quadrics, of dimension $\dim \mathscr X_a=5-2=3$.
	
	We use the Hilbert scheme of lines $\mathrm{Hilb}_{\rm lines}(\mathscr X_a)$ parameterizing all lines $\ell\subset \mathscr X_a$.
Therefore, $\mathrm{Hilb}_{\rm lines}(\mathscr X_a)$ has a $\BF_p$-rational point, which corresponds exactly to a line $\ell\subset \mathscr X_a$ defined over $\BF_p$. This line spans a $2$-dimensional $\BF_p$-subspace of $\BF_q$ on which both trace forms vanish identically, so $\Lambda_a>0$.
	\medskip
	
	\paragraph{Case $m = 5$}
	Here $d=5$.
	The conjugates $a_1,\dots,a_5$ are distinct, and the same diagonalization argument gives $Q_0 = \sum_{i=1}^5 X_i^2$ and $Q_a = \sum_{i=1}^5 a_i X_i^2$. The intersection $\mathscr X_a \subset \mathbb P^4$ is a smooth surface of dimension $2$, which is a del Pezzo surface of degree $4$.
	
	The quadratic form pencil $uQ_0+vQ_a=0$ is singular if and only if its discriminant vanishes, which is equivalent to the condition \[v\cdot a_i+u=0.\]
	
	This corresponds to the projective point
	\[[u:v]=[-a_i:1],\]then there are exactly 5 singular members in the pencil.
	The Frobenius automorphism permutes these 5 conjugate singular members, acting as a cyclic group of order 5, denote by $C_5$.
	
	It is a classical result (Reid) that a smooth del Pezzo surface of degree $4$ contains exactly $16$ lines, these lines are parameterized by the set \[\mathcal{V} = \{ (\epsilon_1,\dots,\epsilon_5) \in \BF_2^5\ : \sum \epsilon_i = 0 \},\] which has size $2^{5-1} = 16$.
	
	Selecting a geometric line as the origin of the affine space, its action can be written as\[z\mapsto Mz+t,\qquad z,t\in \mathcal{V},\] where $M$ is the linear action of $C_5$ on $\mathcal{V}$. Its characteristic polynomial of this action is the $5$-th cyclotomic polynomial \[\Phi_5(T) = T^4+T^3+T^2+T+1\in\BF_2[T].\]  $\Phi_5(1) = 5 \equiv 1 \neq 0$, so the operator $I - M$ is invertible. The fixed point equation
	\[(I-M)z=t\] has exactly one solution. Since the Hilbert scheme of lines is a finite   \'{e}tale  scheme of 16 points, this fixed point corresponds precisely to a $\BF_p$-rational line. Therefore $\Lambda_a = 1$.
	\smallskip
	
	\paragraph{Case $m=4$ and $d=4$}
	Here $\mathscr X_a \subset \mathbb P^3$ is the intersection of two quadrics, which defines a smooth curve of genus $1$  and degree $4$. A line in $\mathbb P^3$ intersects a degree-$4$ curve in at most $4$ points, so $\mathscr X_a$ contains no lines. Consequently, no three points of $X_a=\mathscr X_a(\BF_p)$ are collinear, which implies that the minimum distance of the dual punctured code is at least $4$.
	
	By the Hasse-Weil bound, the number of $\BF_p$-rational points satisfies \[\#\mathscr X_a(\BF_p) \ge p+1 - 2\sqrt{p}>0.\]
	Then we can choose an origin in $X_a=\mathscr X_a(\BF_p)$, and view it as an elliptic curve $E$.
	Set
	\[\mathscr{L} = \mathcal{O}_{E}(1).\]
	Then $h^0(E,\mathscr L)=4$ and $\dim|\mathscr L|=3$. In the group law of $E(\BF_p)$, fixed element $T \in E(\BF_p)$, four points $\{P_1,P_2,P_3,P_4\}$ form a hyperplane section if and only if \[P_1+P_2+P_3+P_4=T.\]
	Here, the complete linear system $\mathscr{L}$ has dimension
	\[h^0(E,\mathscr L)=4,\]
	and since the curve is non-degenerate in $\mathbb P^3$, all effective divisors of degree four linearly equivalent to $\mathscr L$ indeed arise from hyperplane sections.
	
	We now employ an elementary fact from group theory. If the order $N$ of a finite abelian group $G$ satisfies $N\geq5$, then for any $T\in G$, there exist four distinct elements whose sum is $T$.
	
	Therefore, when $s=\#X_a=\#\mathscr X_a(\BF_p)\geq5$, there exist four distinct rational points forming a hyperplane section. The corresponding four columns are linearly dependent, while any three columns are linearly independent, thus $d_{\min}(\C_{R_a}^{\perp})=4$.
	
	If $s\leq4$ and the dual code is non-zero, then $s>\dim(\C_{R_a})$. Since there are no codewords of weight $\leq3$ in the dual code, the only possibility is $s=4$. As all four columns are linearly dependent, the dual distance remains 4. If $s=\dim(\C_{R_a})$, the generator matrix columns are full rank, and the dual code is the zero code. The remaining results can be obtained from Proposition~\ref{pro:dim}.
\end{proof}
\begin{example}
Take $p=5$, $m=6$, and let $a\in\BF_{5^3}\subset\BF_{5^6}$ be a root of
$T^3+2T+1\in\BF_5[T]$. Then $d=3$. An exhaustive computation gives
\[
\#R_a=126,\qquad \Lambda_a=21,
\]
and finds $420$ dependent three-column supports. Hence the program returns
\[
d_{\min}(\C_{R_a}^{\perp})=3,\qquad
A_3(\C_{R_a}^{\perp})=420\cdot4=1680.
\]
\end{example}

\subsection{Divisibility}
By Theorems~\ref{th:weight1} and ~\ref{th:weight3}, we have
\begin{theorem}\label{th:divisible}
	Set
	\[
	\Delta_D=\gcd\{\wt(c):0\ne c\in\C_{D_a}\},\qquad
	\Delta_R=\gcd\{\wt(c):0\ne c\in\C_{R_a}\}.
	\]
	For the zero code, we use the convention $\gcd(\emptyset)=0$.
	Then \[
	\Delta_D=(p-1)\Delta_R.
	\]In particular, for $a\notin\mathbb F_p$.
	\begin{itemize}\item If $m=3$, then \[
		\Delta_R=1,\qquad \Delta_D=p-1.
		\]
		\item If $2\nmid m$ and $m\geq 5$, use $\widehat R_e(\beta)$ as defined above, and set \[
		g_a^{\rm odd}
		=
		\gcd_{\substack{\beta\in K,\ \widehat R_e(\beta)>0\\
				p^{(m-1)/2}-\varepsilon_ms_a(\beta)>0}}
		\left(p^{(m-1)/2}-\varepsilon_ms_a(\beta)\right),
		\]then
		\[
		\Delta_R=p^{(m-5)/2}g_a^{\rm odd},\qquad
		\Delta_D=(p-1)p^{(m-5)/2}g_a^{\rm odd}.
		\]
		\item If $2\mid m$ and $m\geq4$, we take
		\[
		\mathcal B_e^+=\{\beta\in K:\widehat R_e(\beta)>0,
		p^{m/2-1}-\delta(h(a)-Z_a(\beta))>0\},\] and
	\[	g_a^{\rm even}
		=\gcd_{\beta\in\mathcal B_e^+}
		\left(p^{m/2-1}-\delta(h(a)-Z_a(\beta))\right).
		\]then
		\[
		\Delta_R=p^{m/2-2}g_a^{\rm even},\qquad
		\Delta_D=(p-1)p^{m/2-2}g_a^{\rm even}.
		\]
		\item If $m=2$, then $D_a=\emptyset$.
	\end{itemize}
\end{theorem}
The experiment results show that 
\begin{itemize}
\item For $p=3,5$, $m=d=4$, then $g_a^{\rm even}=1$ for all $a$. 
\item For $p=3$, $m=6$ and $d=2,3,6$, then $g_a^{\rm even}=1$ for all $a$.
\end{itemize}
\begin{theorem}
	Assume $D_a\neq\emptyset$. One has $\mathbf 1\notin\C_{D_a}$ and
	\[
	\begin{aligned}
	\mathbf 1\in\C_{R_a}
	&\Longleftrightarrow
	\exists b\in\BF_q:\
	\Tr(br)=1\quad(\forall r\in R_a).
	\end{aligned}
	\]
	Moreover, the  $p$-divisibility of $\C_{D_a}$ and $\C_{R_a}$ is as follows
	\begin{itemize}
		\item If $a\notin\BF_p$ and $m\geq7$ is odd, the code is $p$-divisible.
		\item If $a\notin\BF_p$ and $m\geq6$ is even, the code is $p$-divisible.
	\end{itemize}
	In all automatically $p$-divisible positive-dimension ranges above,
	\[\#R_a\equiv1\pmod p.\]
	Hence $\C_{R_a}$ cannot contain the all-ones vector either.
\end{theorem}

\begin{proof}
	Suppose $\mathbf c_b=\mathbf 1$ and $r\in D_a$. Since $D_a$ is $\BF_p^*$-stable, $2r\in D_a$ as well, and
	\[1=\Tr(b(2r))=2\Tr(br)=2,\]a contradiction. For $\mathbf 1\in \C_{R_a}$, the result is the unwinding of the definition of $\C_{R_a}$.
	
	The divisibility ranges are taken from Theorems~\ref{th:weight1} and ~\ref{th:weight5}. For odd $m$,
	\[\#R_a=\frac{p^{m-2}-1}{p-1}\equiv 1\pmod p.\]
	For even $m\ge6$,
	\[\#R_a=\frac{p^{m-2}-1}{p-1}-\delta p^{m/2-2}h(a)\equiv1\pmod p.\]
	
	If $\C_{R_a}$ is $p$-divisible and contained $\mathbf 1$, its weight $\#R_a$ would be divisible by $p$, this is impossible.
\end{proof}
\begin{remark}\cite[Theorem~11]{21}
	For a linear code over a field of odd characteristic, containment of the all-ones vector together with $p$-divisibility of all non-zero weights implies self-orthogonality. The present theorem shows that this condition is not necessary for the codes arising from this construction.
\end{remark}
\begin{example}
Take $p=3$, $m=6$, and let $a\in\BF_{3^3}\subset\BF_{3^6}$ be a root of
$T^3-T-1\in\BF_3[T]$. Then $d=3$. An exhaustive computation gives
\[
W_{\C_{D_a}}(z)
=1+52z^{54}+312z^{66}+208z^{72}+156z^{78}
\]
and
\[
W_{\C_{R_a}}(z)
=1+52z^{27}+312z^{33}+208z^{36}+156z^{39}.
\]
Thus the program returns
\[
\Delta_D=6,\qquad \Delta_R=3,\qquad \#R_a=52\equiv1\pmod3,
\]
and finds $\mathbf 1\notin\C_{D_a}$ and
$\mathbf 1\notin\C_{R_a}$.
\end{example}
\begin{proposition}
	Let $\mathcal C\subseteq\BF_p^n$ be a nonzero projective linear code of dimension $k_{\mathcal C}$. Then
	\[\gcd\{\wt(c):\ 0\neq c\in \mathcal C\}\mid p^{k_{\mathcal C}-1}.\]
	In particular, the greatest common divisor has no prime divisor other than $p$.
\end{proposition}
\begin{proof}
	Let $G=(g_1,\ldots,g_n)$ with $g_i\in\BF_p^{k_{\mathcal C}}$ be a full-rank generator matrix of $\mathcal C$. For $u\in\BF_p^{k_{\mathcal C}}$, define
	\[w(u)=\#\{i:\ \langle u,g_i\rangle\neq0\}.\]
	Then $w(u)=\sum\limits_{i=1}^n\mathrm{1}_{\{\langle u,g_i\rangle\neq0\}}$ and
	\[\sum_{u\in\BF_p^{k_{\mathcal C}}}w(u)\chi_p(\langle u,x\rangle)=\sum_{i=1}^{n}\sum_{\langle u,g_i\rangle\neq0}\chi_p(\langle u,x\rangle).\]
	
	Now fix $j$ and take $x=\lambda g_j$, $\lambda\in\BF_p^*$. Split the sum into $i=j$ and $i\neq j$.
	\smallskip
	
	\textbf{Case $i\neq j$}.  Projectivity means $g_i$ and $g_j$ are linearly independent, so the map $u\mapsto (\langle u,g_i\rangle,\langle u,g_j\rangle)$ is a surjective linear map to $\BF_p\times\BF_p$.
	
	We count $u$ with $\langle u,g_i\rangle=s\neq0$, for each $s$, the preimage of $(s,t)$ for any $t$ has size $p^{k_{\mathcal C}-2}$. Sum over all $s\neq0$ and all $t$,\[\sum_{s\in\BF_p^*}\sum_{t\in\BF_p}p^{k_{\mathcal C}-2}\cdot\chi_p(\lambda t)=p^{k_{\mathcal C}-2}(p-1)\sum_{t\in\BF_p}\chi_p(\lambda t)=0.\]
	
	\textbf{Case $i= j$}. We restrict to $\langle u,g_j\rangle\neq0$, since $x=\lambda g_j$, $\langle u,x\rangle=\lambda\cdot\langle u,g_j\rangle$. Take $t=\langle u,g_j\rangle$, $t\in\BF_p^*$. Then for each fixed $t\neq0$, the set $\{u:\ \langle u,g_j\rangle=t\}$ is an affine hyperplane, size $p^{k_{\mathcal C}-1}$. Hence,
	\[\sum_{\langle u,g_j\rangle\neq0}\chi_p(\langle u,x\rangle)=\sum_{t\in\BF_p^*}p^{k_{\mathcal C}-1}\chi_p(\lambda t)=-p^{k_{\mathcal C}-1}.\]
	
	Summing over all $i$ gives
	\[\sum_{u\in\BF_p^{k_{\mathcal C}}}w(u)\chi_p(\langle u,x\rangle)=-p^{k_{\mathcal C}-1}.\]
	
	If an integer $\Delta$ divides every $w(u)$, then $p^{k_{\mathcal C}-1}/\Delta$ is a rational algebraic integer, hence belong to $\BZ$. Thus $\Delta\mid p^{k_{\mathcal C}-1}$.
\end{proof}
Application to the punctured code, if $D_a\neq\emptyset$, the punctured code is projective, hence
\[\Delta_R=p^\lambda,\qquad \Delta_D=(p-1)\Delta_R\] for some integer $\lambda\geq0$.

By Theorem~\ref{th:weight1}, for $\beta\neq0$, if $ m\ge4$ is even, we then have
\begin{equation*}
\frac{w_a^{\rm even}(\beta)}{p-1}
=p^{m/2-2}\left(p^{m/2-1}-\delta(h(a)-Z_a(\beta))\right).
\end{equation*}

\begin{lemma}\label{le:021}
	Let $m\ge4$ be even and let $\mathcal B_e^+$ be as defined in Theorem~\ref{th:divisible}.
	Then
	\begin{equation}
	\Delta_R=p^{m/2-2+\kappa_e(a)},
	\end{equation}
	where
	\[
	\kappa_e(a)=
	\begin{cases}
	1,&h(a)-Z_a(\beta)\equiv0\pmod p
	\quad\text{for all }\beta\in\mathcal B_e^+,\\
	0,&\text{otherwise}.
	\end{cases}
	\]
\end{lemma}
\begin{proof}
	By Theorem~\ref{th:weight1}, we know
	\[
	|h(a)-Z_a(\beta)|\leq p+1,
	\qquad \frac{m-2}{2}\ge1.
	\]
	Then the $p$-adic valuation is $0$ if $h(a)-Z_a(\beta)$ cannot be divisible by $p$, and is $1$ if $h(a)-Z_a(\beta)=\pm p$.
\end{proof}
By  Theorem~\ref{th:weight3}, if $m\ge5$ is odd, then
\begin{equation*}
\frac{w_a^{\rm odd}(\beta)}{p-1}
=p^{(m-5)/2}\left(p^{(m-1)/2}-\varepsilon_m s_a(\beta)\right).
\end{equation*}
and we also have $|s_a(\beta)|\le p+1$. Hence the following result can be obtained directly.
\begin{lemma}\label{le:19}
	Assume $m\ge5$ is odd and let $\mathcal B_o$ be as defined in Theorem~\ref{th:weight3}.
	Then
	\begin{equation}
	\Delta_R=p^{(m-5)/2+\kappa_o(a)},
	\end{equation}
	where
	\[
	\kappa_o(a)=
	\begin{cases}
	1,&s_a(\beta)\equiv0\pmod p\quad
	\text{for all }\beta\in\mathcal B_o,\\
	0,&\text{otherwise}.
	\end{cases}
	\]
\end{lemma}
For each $[c]\in\mathbb P(\BF_p+\BF_pa)$, choose a representative $c$ and define once and for all
\[
\ell_c(\beta):=\Tr_{K/\BF_p}(\beta/c),\qquad
\omega_c:=\eta_q(c)=\eta_K(c)^e,\qquad \beta\in K.
\]
\begin{theorem}[Complete classification of maximal divisibility for odd $m$]
	\label{n021-T3}
	Assume $a \notin \BF_p$ and $m \ge 5$ is odd.  Then the following holds.
	
	\begin{enumerate}
		\item[\rm (i)] \textbf{Case $d = 3$.} One has
		\begin{equation}\label{eq:odd-div-d3}
		\Delta_R = p^{(m-3)/2}, \qquad \Delta_D = (p-1)\,p^{(m-3)/2}.
		\end{equation}
		Moreover, $g_a^{\mathrm{odd}} = p$.
		
		\item[\rm (ii)] \textbf{Case $d \ge 5$.} One has
\begin{equation}\label{eq:no1}
		\Delta_R = p^{(m-5)/2}, \qquad \Delta_D = (p-1)\,p^{(m-5)/2}.
\end{equation}
		In particular, $g_a^{\mathrm{odd}} = 1$ whenever $d \ge 5$.
	\end{enumerate}
	
\end{theorem}

\begin{proof}
	\begin{equation}\label{eq:odd-div-correction}
	s_a(\beta)\equiv\sum_{[c] \in \mathbb P(\BF_p + \BF_p a)} \omega_c \, \ell_c(\beta)^{(p-1)/2}\pmod p.
	\end{equation}

	\medskip
	\noindent\textbf{Case $d = 3$.} 
	The projective points $[c^{-1}] \in \mathbb P(K)$ lie on a rational normal curve of degree $d - 1 = 2$. Then
	\[
	N = \frac{p-1}{2}\,(d-1) = p - 1.
	\]
By Theorem~\ref{valsa}, $s_a(\beta)\in\{-p,0,p\}$.	By Theorems~\ref{thm:d=3} and ~\ref{puncode-da},  the weights of punctured code $\C_{R_a}$ are
	\[
	P - B_m, \qquad P, \qquad P + B_m,
	\]
	where
	\[
	P = p^{m-3}, \qquad B_m = p^{(m-3)/2}.
	\]
	Their greatest common divisor is therefore
\begin{equation*}
	\gcd(P - B_m,\, P,\, P + B_m) = \gcd(B_m,\, P) = B_m = p^{(m-3)/2}.
\end{equation*}
	The middle weight $P$ corresponds to the square locus, so the extra factor $p$ survives, $g_a^{\mathrm{odd}} = p$. This establishes \eqref{eq:odd-div-d3}.
	
	\medskip
	\noindent\textbf{Case $d \ge 5$.}
	Euler's criterion gives
	\begin{equation}\label{eq:16}
		s_a(\beta)\equiv
		\sum_{[c]\in\mathbb P(\BF_p+\BF_pa)}
		\omega_c\ell_c(\beta)^{(p-1)/2}
		\pmod p,
	\end{equation} the right-hand side is a homogeneous polynomial on the $d$-dimensional $\BF_p$-vector space $K$.

The projective points
\[
[c^{-1}]\qquad([c]\in\mathbb P(\BF_p+\BF_pa))
\]
	lie on a rational normal curve of degree $d-1$.  After applying the Veronese embedding of degree $(p-1)/2$, the resulting $p+1$ points lie on a rational normal curve of degree
	\[
	N = \frac{p-1}{2}\,(d-1).
	\]
	Since $d\ge5$, 
	\[N \ge 2(p-1) \ge p.\]
	Then the $p + 1$ projective points are linearly independent already in the Veronese embedding. It follows that the polynomial in \eqref{eq:odd-div-correction} is non-zero, and there exists some $\beta \in K$ such that
\begin{equation}\label{eq:17}
	s_a(\beta) \not\equiv 0 \pmod{p}.
\end{equation}

If $e\ge3$, every $\beta\in K$ is attainable, so~\eqref{eq:17} immediately yields 
\[\kappa_o(a)=0.\]

It remains to consider $e=1$, in which case the attainable nonzero values are precisely the squares in $K^*$. If the element $\beta$ in \eqref{eq:17} is a nonsquare, choose a nonsquare $t\in\BF_p^*$. Since $d$ is odd, $t$ is also a nonsquare in $K$, and hence $t\beta$ is a square in $K$. Moreover,
\[s_a(t\beta)=\eta_p(t)s_a(\beta),\]
which is again nonzero modulo $p$. Hence $\kappa_o(a)=0$ also in the full-degree case, and Eq.~\eqref{eq:no1} can be obtained from Lemma~\ref{le:19}.
\end{proof}
\begin{example}
Take $p=3$, $m=9$, and let $a\in\BF_{3^3}\subset\BF_{3^9}$ be a root of
$T^3-T-1\in\BF_3[T]$. Then $d=3$ and $e=3$. An exhaustive computation gives
\[
W_{\C_{D_a}}(z)
=1+3024z^{1404}+13850z^{1458}+2808z^{1512}
\]
and
\[
W_{\C_{R_a}}(z)
=1+3024z^{702}+13850z^{729}+2808z^{756}.
\]
Thus the program returns
\[
\Delta_D=54,\qquad \Delta_R=27,\qquad
g_a^{\mathrm{odd}}=\frac{\Delta_R}{3^{(m-5)/2}}=3.
\]
\end{example}


\begin{theorem}
	\label{n021-T4}
	Let $m \ge 4$ be even, $a \notin \BF_p$, and  $e \ge 2$. Then the following holds.
	
	\begin{enumerate}
		\item[\rm (i)] If $d = 2$ and $e = 2$ (i.e. $m = 4$), then
		\begin{equation}\label{eq:even-div-empty}
		\Delta_R = \Delta_D = 0.
		\end{equation}
		
		\item[\rm (ii)] If $d = 2$ and $e \ge 4$ is even, then
		\begin{equation}\label{eq:even-div-d2-even}
		\Delta_R = p^{m/2-1}, \qquad \Delta_D = (p-1)p^{m/2-1}.
		\end{equation}
		
		\item[\rm (iii)] If $d = 2$ and $e \ge 3$ is odd, then
		\begin{equation}\label{eq:even-div-d2-odd}
		\Delta_R = p^{m/2-2}, \qquad \Delta_D = (p-1)p^{m/2-2}.
		\end{equation}
		
		\item[\rm (iv)] If $d \ge 3$, regardless of whether $e = 2$ or $e \ge 3$,
		\begin{equation}\label{eq:even-div-dge3}
		\Delta_R = p^{m/2-2}, \qquad \Delta_D = (p-1)p^{m/2-2}.
		\end{equation}
	\end{enumerate}
\end{theorem}

\begin{proof} (i) By Proposition~\ref{pro:dim}, then $D_a = \emptyset$ and Eq.~\eqref{eq:even-div-empty} holds.
	\smallskip
	
	As a $\BF_p$-valued function,
	\[
	\mathrm{1}_{\{\ell_c(\beta) \neq 0\}} = \ell_c(\beta)^{p-1},
	\]
	so the normalized weight correction is
	\begin{equation}\label{eq:even-div-correction}
	h(a) - Z_a(\beta) \equiv \sum_{[c] \in \mathbb P(\BF_p + \BF_pa)} \omega_c \, \ell_c(\beta)^{p-1} \pmod{p}.
	\end{equation}

	\noindent\textbf{Case $d \ge 3$.} We have
	\[
N = (p-1)(d - 1)
	\ge 2(p-1) \ge p,
	\]
	so the $p + 1$ points are linearly independent after the Veronese embedding. In both cases \eqref{eq:even-div-correction} is not the zero polynomial, and there exists $0\ne \beta \in K$ with
	\begin{equation}\label{eq:even-div-nonzero}
	h(a)-Z_a(\beta) \not\equiv 0 \pmod{p}.
	\end{equation}
	
	In the present even-$m$ case with $e\ge2$, Lemma~\ref{bb} shows that the relevant nonzero values of $\beta$ are attained by the map $b\mapsto \Tr_{\BF_q/K}(b^2)$. Hence \eqref{eq:even-div-nonzero} gives $\kappa_e(a)=0$. Combined with Lemma~\ref{le:021}, \eqref{eq:even-div-dge3} follows.
	
	\medskip
	\noindent\textbf{Case $d = 2$.}
	Here $\mathbb P(\BF_p + \BF_pa) = \mathbb P(K)$. If $e$ is even, $\omega_c = 1$ for all $c\neq0$. For each non-zero $\beta$, the linear functional $c^{-1} \mapsto \Tr_{K/\BF_p}(\beta / c)$ vanishes at exactly one projective point, so
	\[
	h(a) = p + 1, \qquad Z_a(\beta) = 1,
	\]
	and therefore
\begin{equation*}
	h(a) - Z_a(\beta) = p.
\end{equation*}
	For $\beta = 0$, $h(a)=Z_a(0)$. Hence the extra g.c.d. of all normalized weights is exactly $p$, $e = 2$ is the empty-code case by Lemma~\ref{n0}, while $e \ge 4$ yields \eqref{eq:even-div-d2-even}.
	
	If $e$ is odd, Lemma~\ref{n0} gives $h(a) = 0$. Each $\beta \neq 0$ still lies on one zero hyperplane, so $Z_a(\beta) = \pm 1$ and $h(a) - Z_a(\beta) = \mp 1$, thus $\kappa_e(a) = 0$, yielding \eqref{eq:even-div-d2-odd}.
\end{proof}
\begin{example}
Take $p=3$, $m=6$, and let $a\in\BF_{3^3}\subset\BF_{3^6}$ be a root of
$T^3-T-1\in\BF_3[T]$. Then $d=3$ and $e=2$. An exhaustive computation gives
\[
W_{\C_{D_a}}(z)
=1+52z^{54}+312z^{66}+208z^{72}+156z^{78}
\]
and
\[
W_{\C_{R_a}}(z)
=1+52z^{27}+312z^{33}+208z^{36}+156z^{39}.
\]
Thus the program returns
\[
\Delta_D=6,\qquad \Delta_R=3.
\]
\end{example}


\section{Conclusions}

We constructed a family of \(p\)-ary trace codes from two quadratic trace
conditions and determined their parameters and complete weight distributions
using quadratic forms and character sums. The results explicitly describe
the roles of \(d=[\mathbb F_p(a):\mathbb F_p]\), \(e=m/d\), and their parities.
By exploiting the \(\mathbb F_p^*\)-orbit structure, we also obtained punctured
codes and characterized their dual distances, self-orthogonality,
divisibility and projectivity. For \(p>3\), weighted
self-orthogonality was reduced to the condition
\(\ker L_a\cap \mathcal S_p^{X_a}\ne\emptyset\), with explicit existence ranges.

Compared with \cite{10, 11, 15, 20, 21}, the main distinction of this paper lies in the underlying
construction and its scope. In particular, Xiang and Tang \cite{15} treats the double-trace square-function construction only for $p = 3$ and even $m$, whereas we extend it to arbitrary odd primes and both parities
of m. Unlike \cite{10, 11, 21}, which employ bent functions, special defining sets or $p$-divisibility as
the main construction framework, we study trace codes arising from a system of two quadratic
trace equations. The general counting results in \cite{20} serve as a useful tool in our analysis.
Beyond the weight distributions, we also study the orbit-induced punctured codes and their
structural properties.

Further study may consider additional quadratic trace constraints and related
quantum or LCD codes.

\section*{Data Availability}

No data was used for the research described in the article.

\Acknowledgements{This research is supported by the National Natural Science Foundation of China under Project 62402004, 62201009, the Natural Science Foundation of Anhui Province under Project 2508085QA033, the Natural
Science Research Project of Anhui Provincial Education Committee under Project 2024AH050430 and in part by the Talent Research Foundation of Anhui Agricultural University, China under Grant rc482402.}

\end{document}